\AtBeginDocument{%
  \paperwidth=\dimexpr
    1in + \oddsidemargin
    + \textwidth
    + 1in + \oddsidemargin
  \relax
  \paperheight=\dimexpr
    1in + \topmargin
    + \headheight + \headsep
    + \textheight
    + 1in + \topmargin
  \relax
  \usepackage[pass]{geometry}\relax
}

\RequirePackage{fix-cm}
\documentclass[smallcondensed]{svjour3}     % onecolumn (ditto)
\smartqed % flush right qed marks, e.g. at end of proof
\usepackage{graphicx}
\usepackage{mathptmx}      % use Times fonts if available on your TeX system
\usepackage{amsmath}
\usepackage{amssymb}
\usepackage{mathtools}
\usepackage{comment}
\usepackage{enumitem}
\usepackage{xcolor}
\usepackage{mathrsfs}
\usepackage{float}
\usepackage[scr=mma]{mathalpha}
\usepackage{bm}
\usepackage{amsbsy}
\usepackage{dsfont}

\journalname{}

\begin{document}

\title{Prohorov Metric-Based Nonparametric Estimation of the Distribution of Random Parameters in Abstract Parabolic Systems with Application to the Transdermal Transport of Alcohol%\thanks{Grants or other notes
}

\titlerunning{Prohorov Metric-Based Nonparametric Estimation of Random Parameters}        % if too long for running head

\author{Lernik Asserian$^1$         \and
        Susan E. Luczak$^{2,4}$ \and I.G. Rosen$^{3,4}$ %etc.
}

\authorrunning{Lernik Asserian et al.} % if too long for running head

\institute{L. Asserian \\
              \email{lernik@stanford.edu}            \\
              Susan E. Luczak \\
              \email{luczak@usc.edu} 
              \\
              I.G. Rosen \\
              \email{grosen@usc.edu} \\
            $^1$ Department of Mathematics, Stanford University, Stanford, CA, USA
            \\
            $^2$ Department of Psychology, University of Southern California, Los Angeles, CA, USA \\
            $^3$ Modeling and Simulation Laboratory, Department of Mathematics, University of Southern California, Los Angeles, CA, USA
\\
              $^4$ This study was funded in part by the National Institute on Alcohol Abuse and Alcoholism under Grant Numbers: R21AA017711 and R01AA026368, S.E.L. and I.G.R.
         \\ }


\maketitle

\begin{abstract}
We consider a Prohorov metric-based nonparametric approach to estimating the probability distribution of a random parameter vector in discrete-time abstract parabolic systems. We establish the existence and consistency of a least squares estimator. We develop a finite-dimensional approximation and convergence theory, and obtain numerical results by applying the nonparametric estimation approach and the finite-dimensional approximation framework to a problem involving an alcohol biosensor, wherein we estimate the probability distribution of random parameters in a parabolic PDE. To show the convergence of the estimated distribution to the ``true" distribution, we simulate data from the ``true" distribution, apply our algorithm, and obtain the estimated cumulative distribution function. We then use the Markov Chain Monte Carlo Metropolis Algorithm to generate random samples from the estimated distribution, and perform a generalized (2-dimensional) two-sample Kolmogorov-Smirnov test with null hypothesis that our generated random samples from the estimated distribution and generated random samples from the ``true" distribution are drawn from the same distribution. We then apply our algorithm to actual human subject data from the alcohol biosensor and observe the behavior of the normalized root-mean-square error (NRMSE) using leave-one-out cross-validation (LOOCV) under different model complexities.
\keywords{Nonparametric estimation \and Prohorov metric \and Existence and consistency of estimators \and Least squares estimation \and Random discrete time dynamical systems \and Random partial differential equations \and Finite dimensional approximation and convergence \and Alcohol biosensor \and Transdermal alcohol concentration \and Breath alcohol concentration}
%\subclass{35K90 \and 62G05 \and 65C20 \and 92C47 \and 92C50}
\end{abstract}

\section{Introduction}
\label{intro}

In the absence of a blood sample, breath alcohol concentration (BrAC) obtained by a breathalyzer is typically used as a surrogate for blood alcohol concentration (BAC) for measuring the alcohol level in the human body. This is the case in law enforcement, medical research, and clinical therapy. The breathalyzer was developed by Borkenstein based on a redox (i.e. oxidation and reduction) reaction and Henry's law \cite{Labianca:1990}. While not entirely without controversy, there is reasonable agreement between BrAC and BAC, and this relationship has been shown to be relatively consistent across individuals under varying conditions \cite{Labianca:1990}. Even though there is widespread acceptance among researchers and clinicians of using BrAC as a proxy for BAC, collecting and using samples of BrAC can be both challenging and limiting. For example, it is impractical to collect continuous or near-continuous BrAC data. Also, blowing into the breathalyzer so that a deep lung sample is obtained can be difficult, and if still drinking, the blown sample may be contaminated by mouth alcohol. Alternatively, ethanol, the type of alcohol contained in alcoholic beverages, is highly miscible and finds its way into all the water in the body. Moreover, it has been known for a long time that the alcohol contained in perspiration is positively correlated with BAC (and BrAC) \cite{Swift:2000}, but unfortunately, due to a number of confounding factors, the actual functional relationship between alcohol in perspiration and BAC/BrAC is not entirely clear. The technology required to measure the alcohol level in perspiration has been evolving and there are now biosensors available that measure transdermal alcohol concentration (TAC). For the most part, TAC sensors also rely on a redox reaction in the form of a fuel cell that produces 4 electrons for each molecule of ethanol \cite{Marques:2009}. The result is an electrical current that can be measured and converted into TAC. However, for the device to be used by researchers and clinicians as a meaningful quantitative indicator of alcohol in the blood, a means of converting TAC into BrAC or BAC is required.

TAC is the amount of alcohol that diffuses from the dermal layer through the epidermal layer. The dermal layer of the skin is the layer that has active blood supply, and the epidermal layer of the skin is the layer that does not contain blood. After consuming alcohol, the liver metabolizes most of the alcohol, but some of the alcohol leaves through exhaled breath, some through urine, and about $1\%$ diffuses through the skin in the form of perspiration and sweat \cite{Sakai:2006}. A number of biosensors have been developed over the past few decades for measuring TAC -- among them are the WrisTAS$^{\text{TM}}$7 (Figure \ref{Fig.1}, left) manufactured by Giner, Inc. in Waltham, MA and the SCRAM CAM$^{\textregistered}$ (Secure Continuous Remote Alcohol Monitor) (Figure \ref{Fig.1}, right) by Alcohol Monitoring Systems, Inc. (AMS) in Littleton, CO.

\begin{figure}[H]
\centering
\includegraphics[width=3.9cm ,height= 3.2cm]{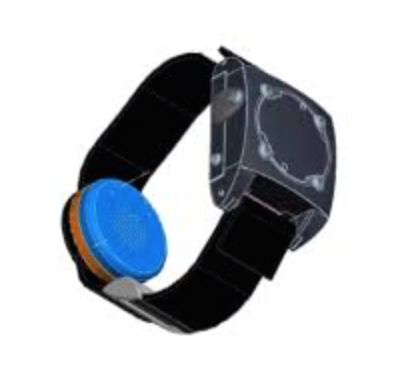}
\hspace{0.5 in}
\includegraphics[width=4.2cm ,height= 2.6cm]{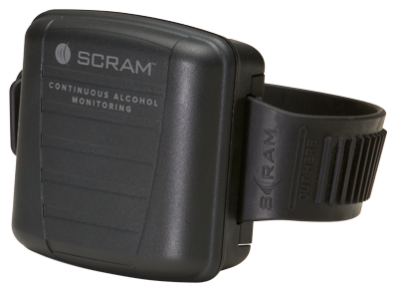}
\caption{WrisTAS$^{\text{TM}}$7 (left) and SCRAM CAM$^{\textregistered}$ (right) alcohol biosensors}
\label{Fig.1}
\end{figure}

Furthermore, there are variations in TAC 1) between sensors including the same brand of sensors, 2) between persons due to thickness of the skin, porosity (the quality of being porous, i.e. having tiny holes that fluid flows through), tortuosity (i.e. the diffusion and fluid flow through turns and twists in a porous media), and diffusion rate through the skin, and 3) within person across drinking episodes due to environmental conditions, temperature and humidity, and skin hydration and vasodilation (i.e. the dilation of blood vessels which can result in a lower blood pressure). Taken together, these variations create significant challenges in obtaining a direct method of converting TAC to BrAC/BAC.

In previous research, we used deterministic models to address these types of issues \cite{Banks:1997,Banks:1989}. For the forward model, where the input is BrAC and the output is TAC, in \cite{Dai:2016,Rosen:2014}, we developed a deterministic approach wherein a model based on a one-dimensional diffusion equation using Fick's law of diffusion \cite{Smith:2004} was calibrated to breath analyzer BrAC data and TAC data for a single drinking episode for each individual subject. The resulting fit of the model was then used to estimate BrAC for other drinking episodes by that individual wearing the device when only TAC (and not BrAC) is obtained. However, the calibration process, which requires the simultaneous collection of BrAC and TAC for at least one drinking episode, is difficult to conduct, time consuming, and often impractical. In \cite{Sirlanci:2019(1),Sirlanci:2017,Sirlanci:2019(2)}, we eliminated the calibration process to make the process of estimating BrAC from TAC simpler for researchers and clinicians. In this series of papers, we allowed the parameters in the diffusion equation to be random and then, rather than fitting the parameters themselves, we fit the distribution of the parameters to population data and thus obtained a population model. This model could then be used to obtain an estimate of BrAC from TAC along with confidence bands for an individual subject without first requiring individualized calibration.

As in \cite{Sirlanci:2019(1),Sirlanci:2017,Sirlanci:2019(2)}, we consider a physics-based model for the transdermal transport of ethanol in the form of a random parabolic partial differential equation. The parameters of this model are considered to be random due to a number of different sources of uncertainty. The first task is to estimate the distribution of these random parameters, which we consider here. 

Converting to dimensionless variables, the following initial-boundary value problem is obtained.
\begin{align}
    \frac{\partial x}{\partial t}(t,\eta) &= q_1 \frac{\partial^2 x}{\partial\eta^2}(t,\eta), \quad 0<\eta<1, \quad t>0, \label{eq.1}\\
    q_1\frac{\partial x}{\partial \eta}(t,0) &= x(t,0), \quad t>0, \label{eq.2}\\
    q_1\frac{\partial x}{\partial \eta}(t,1) &= q_2 u(t), \quad t>0, \label{eq.3}\\
    x(0,\eta) &= x_0, \quad 0<\eta<1, \label{eq.4}\\
    y(t) &= x(t,0), \quad t>0, \label{eq.5}
\end{align}
where $t$ is the temporal variable and $\eta$ is the spatial variable, $x(t,\eta)$ is the concentration of ethanol in the epidermal layer of the skin at time $t$ and depth $\eta$, and $\eta=0$ at the skin surface and $\eta=1$ at the boundary between the epidermal and dermal layers of the skin. We denote BrAC/BAC at time $t$ by $u(t)$, and TAC at time $t$ by $y(t)$. Equation (\ref{eq.1}) describes the transport of ethanol through the epidermal layer of the skin with boundary conditions (\ref{eq.2}) and (\ref{eq.3}) representing the evaporation of ethanol at the skin surface and flux of ethanol across the boundary between the epidermal and dermal layers of the skin, respectively. We assume that the initial condition (\ref{eq.4}) is $x_0(\eta)=0$, $0<\eta<1$, representing the assumption that at time $t=0$ there is no alcohol in the epidermal layer of the skin. Finally, the output equation (\ref{eq.5}) is the TAC level at the skin surface measured by the biosensor. The random parameters in these equations are $q_1$, which represents the normalized diffusivity, and $q_2$, which represents the normalized flux gain at the boundary between the dermal and epidermal layers of the skin. These parameters depend on the specific sensor that is being used in the experiment, the environmental conditions, and the individual wearing the sensor. The parameter vector is $\bm{q} = (q_1,q_2) \in Q$, where $Q$ is a compact subset of $\mathbb{R}^+ \times \mathbb{R}^+$ with metric $d_Q$.

The fundamental idea in both our earlier work, \cite{Sirlanci:2019(1),Sirlanci:2017,Sirlanci:2019(2)}, and our treatment here is to account for the uncertainty in the model parameters by fitting a population model to observations of both BrAC and TAC for a particular sample of subjects.  Fitting the population model consists of estimating the distribution of the random parameters.  Once this is done, the resulting population model can then be used to estimate BrAC from a subject's TAC when no observations of BrAC are available.  The underlying statistical model used to estimate the distribution of the parameters is based on the assumption that the dynamics in this physics-based model is common to the entire sample (i.e. all individuals who wore different sensors under different environmental conditions and all molecules of ethanol that were transported from the blood through the skin to the sensor) and then rely on randomness and uncertainty in the model parameters to capture the un-modeled variation. Thus, we assume that each data point is an observation that consists of the mean behavior plus a random error. 

In \cite{Sirlanci:2018}, abstract discrete-time random parabolic systems with unbounded input and output operators (of which (\ref{eq.1})-(\ref{eq.5}) is a particular example) are considered.  Using the ideas developed in \cite{Gittelson:2014}, the systems of this general form are reformulated in weak form in a Gelfand triple setting wherein the corresponding Hilbert spaces are Bochner spaces with measures derived from the unknown distribution of the random parameters. In this way, the random parameters are treated much like additional spatial variables.  Finite dimensional approximation is achieved via Galerkin approximation with convergence arguments based on linear semigroup theory. If we let $N$ denote the level of discretization for the spatial variable $\eta$ and the random parameter vector $\bm{q}$, using linear splines and characteristic functions, respectively, we obtain matrix representations for the approximating linear operators, $\bm{A}^N, \bm{B}^N, \text{ and } \bm{C}^N$.  Then using the variation of parameters formula, the solution to the underlying approximating discrete-time dynamical system then takes the form of a convolution as 
\begin{align*}
    x^N_k &= \sum_{j=0}^{k-1} e^{\bm{A}^N (k-j-1)\tau} \bm{B}^N u_j, \quad k= 1,2, \dots,
\end{align*}
with the approximate output, or TAC, given by
\begin{align*}
    y^N_k &= \sum_{j=0}^{k-1} h_{k-j} u_j, \quad k=1,2, \dots.
\end{align*}
where the convolution kernel or filter $h_k = \bm{C}^Ne^{\bm{A}^N (k-1)\tau}\bm{B}^N$, $k=1,2, \dots$, and $\tau$ denotes the length of the sampling interval. 

In the sequence of papers \cite{Sirlanci:2019(1),Sirlanci:2019(2),Sirlanci:2018}, we made the assumption that the distribution of the random parameters is absolutely continuous with density parameterized by a parameter $\rho$ (e.g., a truncated bivariate Normal), estimated via least squares as $\rho^N$. We argued that there exists a subsequence $\{\rho^{N_j}\}$ that converges to the least squares estimate that would be obtained based on the original infinite-dimensional dynamical system. 

Our approach here is to not make any assumptions about the form of the distribution of the random parameters. We use a nonparametric estimation approach to estimate the shape of the cumulative distribution function directly rather than estimating parameters that determine a density, thus we are not limiting ourselves to any specific distribution type. We establish the existence and consistency of a least squares estimator and develop a finite-dimensional approximation and convergence theory. Our approach is an extension of the approach taken by Banks and his coauthors in \cite{Banks:2018} and \cite{Banks:2012}. Banks' results are for continuous-time dynamical systems with a single subject. We are interested in discrete-time dynamical systems with multiple subjects.

An outline of the remainder of the paper is as follows. In section \ref{sec:2}, we set up the mathematical model for our discrete-time dynamical systems and the associated statistical model. In section \ref{sec:3}, we establish the existence and consistency of the least squares estimator. In section \ref{sec:4}, we establish a convergence theory for finite-dimensional approximations to the estimator. In section \ref{sec:5}, we consider abstract parabolic systems in a Gelfand triple setting and the application of our abstract estimation and approximation framework to them. Finally, in section \ref{sec:6}, we present the application of our scheme to the transdermal transport of alcohol, first in the context of an example involving simulated data, and then in two examples involving human subject data collected in the laboratory of one of the co-authors (SEL) utilizing two different TAC biosensor devices. In the simulation case in section \ref{sec:6.1}, in order to show the convergence of the estimated distribution of the parameter vector $\bm{q}=(q_1,q_2)$ to the ``true" distribution, we take BrAC inputs from different drinking episodes, and simulate the TAC outputs using a specific distribution as the ``true" distribution. We estimate the distribution of the parameter vector using the algorithm from section \ref{sec:4}. Then, we generate samples from the estimated distribution using the Markov Chain Monte Carlo (MCMC) Metropolis algorithm. We perform a generalized (2-dimensional) two-sample Kolmogorov-Smirnov test (KS-test) in order to test the null hypothesis that our generated samples from the estimated distribution and the generated samples from the ``true" distribution are drawn from the same distribution. For the human subject data in section \ref{sec:6.2}, we take drinking episodes from nine individuals, and apply the leave-one-out cross-validation (LOOCV) method. Each time, we estimate the distribution of the parameter vector $\bm{q}$ using the training set. We simulate $100$ TACs using the BrAC input of the test set and $100$ samples of $\bm{q}$ generated from the estimated distribution. Using the average of the TACs as an estimation of the ``true" TAC, we observe the behavior of the normalized root-mean-square error (NRMSE) under different model complexities and consider their computational costs.

\section{The Mathematical Model}
\label{sec:2}

Analogously to what was done in \cite{Banks:2018,Banks:2012}, consider the following discrete-time mathematical model for the $i^{th}$ subject at time-step $k$
\begin{align*}
    x_{k,i} &= g_{k-1}(x_{k-1,i},u_{k-1,i};\bm{q}), \enskip k = 1,\dots,n_i, \enskip i = 1,\dots,m,\\
    x_{0,i} &= \phi_{0,i}, \enskip i = 1,\dots,m,
\end{align*}
where $\bm{q} \in Q$ is the parameter vector, $Q$ is the set of admissible parameters, $g_{k-1}: \mathcal{H} \times \mathbb{R}^{\nu} \times Q \to \mathcal{H}$, $\mathcal{H}$ is, in general, an infinite-dimensional Hilbert space, and $u_{k-1,i} \in \mathbb{R}^{\nu}$ is the input. The output is given by
\begin{align*}
    y_{k,i} = h_k(x_{k,i}, \phi_{0,i}, u_{k,i}; \bm{q}), \enskip k = 1,\dots,n_i, \enskip i = 1,\dots,m,
\end{align*}
where $h_k: \mathcal{H} \times \mathcal{H} \times \mathbb{R}^{\nu} \times Q \to \mathbb{R}$.

We consider the aggregate problem of observing $n_i \times m$ sampled measurements of the mean behavior of the output plus a random error. Define
\begin{align}
    Y_{k,i} = \bar{y}_{k,i}(P_0) + e_{k,i}, \enskip k = 1,\dots,n_i, \enskip i = 1,\dots,m, \label{eq.6}
\end{align}
where
\begin{align*}
    \bar{y}_{k,i}(P) = E[h_k(x_{k,i}, \phi_{0,i}, u_{k,i}; \bm{q})] = \int_Q h_k(x_{k,i}, \phi_{0,i}, u_{k,i}; \bm{q}) dP, 
\end{align*}
and $e_{k,i}$ are independent and identically distributed (i.i.d) random variables with mean $0$ and common variance $\sigma^2$, and $P$ is a probability measure on the Borel sigma algebra on $Q$. Let $\mathcal{P}(Q)$ denote the set of all probability measures defined on the Borel sigma algebra on $Q$, $\Sigma_Q$, and let $P_0$ denote the ``true" distribution of the random vector $\bm{q}$. The goal is to find an estimate of $P_0$. In order to generate an estimator for $P_0$, we employ an abstract framework based on nonlinear least squares and the Prohorov metric on $\mathcal{P}(Q)$ to establish theoretical results and computational tools. 

In \cite{Banks:2012}, Banks and his co-authors developed a Prohorov metric-based framework for estimation of the probability measure for random parameters in continuous-time dynamical systems. We briefly describe the Prohorov metric and its properties. Let $Q$ be a Hausdorff metric space with metric $d$. Define $$C_b(Q) = \{f: Q \to \mathbb{R} \enskip | \enskip f \enskip \text{is bounded and continuous}\},$$ and given any probability measure $P \in \mathcal{P}(Q)$ and some $\epsilon > 0$, an $\epsilon$-neighborhood of $P$ is defined by
    \begin{align*}
        B_\epsilon(P) = \bigg{\{} \tilde{P} \enskip \bigg{|} \enskip \Bigg{|} \int_Q f(\bm{q}) d\tilde{P}(\bm{q}) - \int_Q f(\bm{q}) dP(\bm{q}) \Bigg{|} < \epsilon, \enskip \text{for all} \enskip f \in C_b(Q) \bigg{\}}.
    \end{align*}
The Prohorov metric $\rho$ on $\mathcal{P}(Q) \times \mathcal{P}(Q)$ is then defined so that given two probability measures, $P$ and $\tilde{P}$ in $\mathcal{P}(Q)$, $\tilde{P} \in B_\epsilon(P)$ if and only if $\rho(P,\tilde{P}) < \epsilon$. For any two measures $P, \tilde{P} \in \mathcal{P}(Q)$, the Prohorov metric $\rho$ is defined as
    \begin{align*}
        \rho(P,\tilde{P}) = \inf \{ \epsilon > 0 \enskip | \enskip \tilde{P}(E) \leq P(E^\epsilon) + \epsilon \enskip \text{and} \enskip P(E) \leq \tilde{P}(E^\epsilon) + \epsilon, \enskip \text{for all} \enskip E \in \Sigma_Q \},
    \end{align*}
where $E^\epsilon$, the $\epsilon$-neighborhood of $E$ is defined by $$E^\epsilon = \{ \tilde{\bm{q}} \in Q \enskip | \enskip d(\tilde{\bm{q}},E) < \epsilon \} = \{ \tilde{\bm{q}} \in Q \enskip | \enskip \inf_{\bm{q} \in E} d(\bm{q},\tilde{\bm{q}}) < \epsilon \},$$ for all nonempty $E \in \Sigma_Q$, the Borel sigma algebra on $Q$. 

It can be shown that $(\mathcal{P}(Q),\rho)$ is in fact a metric space. Moreover, given a sequence of measures $P_M \in \mathcal{P}(Q)$ for all $M=1,\dots,\infty$, and $P \in \mathcal{P}(Q)$, we say $P_M$ converges weakly to $P$, $P_M \xrightarrow{w^*} P$, if and only if $\rho(P_M,P) \to 0$; the Prohorov metric metrizes the weak convergence of measures. It is important to note that the weak$^*$ topology and the weak topology are equivalent on the space of probability measures.

The metric space $(\mathcal{P}(Q),\rho)$ has properties that will be especially relevant in what follows. Let $D = \{\delta_{\bm{q}} \enskip | \enskip \bm{q} \in Q \},$ be the space of Dirac measures on $Q$, where for all $E \in \Sigma_Q$,
\begin{align*}
    \delta_{\bm{q}}(E) =
    \begin{dcases}
    \enskip 1 & \enskip \text{if } \bm{q} \in E \\
    \enskip 0 & \enskip \text{if } \bm{q} \notin E
    \end{dcases}
\end{align*}
then if $\bm{q}_1,\bm{q}_2 \in Q$, $\rho(\delta_{\bm{q}_1},\delta_{\bm{q}_2}) = \min\{d(\bm{q}_1,\bm{q}_2),1\}$. The metric space $(\mathcal{P}(Q),\rho)$ is separable if and only if the metric space $(Q,d)$ is separable, and in this case, the sequence $\{\bm{q}_j\}_{j=1}^\infty$ is Cauchy in $(Q,d)$ if and only if the sequence $\{\delta_{\bm{q}_j}\}_{j=1}^\infty$ is Cauchy in $(\mathcal{P}(Q),\rho)$. We also have $(Q,d)$ is complete if and only if $(\mathcal{P}(Q),\rho)$ is complete, and $(Q,d)$ is compact if and only if $(\mathcal{P}(Q),\rho)$ is compact. The details and proofs can be found in \cite{Banks:2012}.

Assume the metric space $(Q,d)$ is separable and let $Q_d = \{\bm{q}_j\}_{j=1}^{\infty}$ be a countable dense subset of $Q$. Define the dense (see \cite{Banks:2012}) subset of $\mathcal{P}(Q)$, $\tilde{\mathcal{P}}_d(Q)$, as 
\begin{align*}
    \tilde{\mathcal{P}}_d(Q) = \{P \in \mathcal{P}(Q) \enskip | \enskip P=\sum_{j=1}^{M} p_j \delta_{\bm{q}_j}, \bm{q}_j \in Q_d, M \in \mathbb{N}, p_j \in [0,1] \cap \mathbb{Q}, \sum_{j=1}^{M} p_j =1\},
\end{align*}
the collection of all convex combinations of Dirac measures on $Q$ with nodes $\bm{q}_j \in Q_d$ and rational weights $p_j$, and for each $M \in \mathbb{N}$ let
\begin{align*}
        \mathcal{P}_M(Q) = \{P \in \tilde{\mathcal{P}}_d(Q) \enskip | \enskip P=\sum_{j=1}^{M} p_j \delta_{\bm{q}_j}, \bm{q}_j \in \{\bm{q}_j\}_{j=1}^{M}\}.
\end{align*}
Let $\bm{n}=\{n_i\}_{i=1}^m$ and define
\begin{align*}
    J_{\bm{n},m}(\bm{Y};P) = \sum_{i=1}^m \sum_{k=1}^{n_i} (Y_{k,i} - \bar{y}_{k,i}(P))^2,
\end{align*}
where $\bm{Y} = (\{Y_{k,i}\}_{k=1}^{n_i})_{i=1}^{m}$ and define the estimator
\begin{align}
    P_{\bm{n},m} = \arg \min_{P \in \mathcal{P}(Q)} J_{\bm{n},m}(\bm{Y};P).
    \label{eq.7}
\end{align}

Let $\mathscr{y}_{k,i}$ be realizations of the random variables $Y_{k,i}$, and define
\begin{align}
    \hat{P}_{\bm{n},m} = \arg \min_{P \in \mathcal{P}(Q)} J_{\bm{n},m}(\pmb{\mathscr{y}};P) = \arg \min_{P \in \mathcal{P}(Q)} \sum_{i=1}^m \sum_{k=1}^{n_i} (\mathscr{y}_{k,i} - \bar{y}_{k,i}(P))^2,
    \label{eq.8}
\end{align}
where $\pmb{\mathscr{y}} = \{\{\mathscr{y}_{k,i}\}_{k=1}^{n_i}\}_{i=1}^{m}$. Now we cannot exactly compute $\hat{P}_{\bm{n},m}$ since $\bar{y}_{k,i}(P)$ typically must be approximated numerically. Let $\bar{y}^N_{k,i}(P)$ be an approximation of $\bar{y}_{k,i}(P)$ based, for example, on a Galerkin scheme with $N$ denoting the level of discretization. In addition, we define our approximating estimator over the set $\mathcal{P}_M(Q)$ where $M$ denotes the number of nodes, $\{\bm{q}_j\}_{j=1}^M$, so that the least squares optimization is now over a finite set of parameters, namely the $\{p_j\}_{j=1}^M$. Our approximating estimator is then given by
\begin{align*}
        \hat{P}_{\bm{n},m,M}^N = \arg \min_{P \in \mathcal{P}_M(Q)} J_{\bm{n},m}^N(\pmb{\mathscr{y}},P) = \arg \min_{P \in \mathcal{P}_M(Q)} \sum_{i=1}^m \sum_{k=1}^{n_i} (\mathscr{y}_{k,i} - \bar{y}^N_{k,i}(P))^2.
\end{align*}

\section{Existence and Consistency of the Least Squares Estimator}
\label{sec:3}

In order to establish the existence of the estimator $P_{\bm{n},m}$ in (\ref{eq.7}), it is sufficient to show the existence of the estimator $\hat{P}_{\bm{n},m}$ in (\ref{eq.8}). Recall that the estimator $\hat{P}_{\bm{n},m}$ is obtained from the realizations $\{\mathscr{y}_{k,i}\}, \enskip k=1,\dots,n_i, \enskip i=1,\dots,m$ of the random variables $\{Y_{k,i}\}, \enskip k=1,\dots,n_i, \enskip i=1,\dots,m$. The following theorem establishes the existence of the estimator $\hat{P}_{\bm{n},m}$.

\begin{theorem}
    For $i=1,2,\dots,m$, let $J_{n_i} : \mathbb{R}^{n_i} \times \mathcal{P}(Q) \to \mathbb{R}$ be defined by
    \begin{align*}
        J_{n_i}(\pmb{\mathscr{y}}_i;P) = \sum_{k=1}^{n_i} (\mathscr{y}_{k,i} - \bar{y}_{k,i}(P))^2,
    \end{align*}
     where $\pmb{\mathscr{y}}_i = [\mathscr{y}_{1,i},\dots,\mathscr{y}_{n_i,i}]^T$ and set
    \begin{align*}
        J_{\bm{n},m}(\pmb{\mathscr{y}};P) = \sum_{i=1}^m  J_{n_i}(\pmb{\mathscr{y}}_i;P).
    \end{align*}
    Assume $(Q,d)$ is separable and compact. Let $(\mathcal{P}(Q),\rho)$ be the space of probability measures $\mathcal{P}(Q)$ with the Prohorov metric $\rho$. Assume that for each $P \in \mathcal{P}(Q)$ we have a measurable function $J_{n_i}( . ;P) : \mathbb{R}^{n_i} \to \mathbb{R}$, and also for each $\pmb{\mathscr{y}}_i$, we have a continuous function $J_{n_i}(\pmb{\mathscr{y}}_i;.) : \mathcal{P}(Q) \to \mathbb{R}$.
    Then there exists a measurable function $\hat{P}_{\bm{n},m} : \prod_{i=1}^m \mathbb{R}^{n_i} \to \mathcal{P}(Q)$ such that $$J_{\bm{n},m}(\pmb{\mathscr{y}};\hat{P}_{\bm{n},m}) = \inf_{P \in \mathcal{P}(Q)} J_{\bm{n},m}(\pmb{\mathscr{y}};P).$$
\end{theorem}

\begin{proof}
    The theorem can be proved in a similar way as in \cite{Banks:2012}.
    \qed
\end{proof}

Next, we establish the consistency of the estimator $P_{\bm{n},m}$ by showing that $P_{\bm{n},m}$ converges to the ``true" distribution, $P_0$, in probability.

Consider the following assumptions for a fixed $m$,
\begin{itemize}[leftmargin=0.25in]
    \item[\textbf{A0.}] Let $n=\max n_i$, and let $y_{k,i} = 0$ for $n_i < k \leq n, \enskip i=1,\dots,m$.
    
    \item[\textbf{A1.}] For any fixed $n$ and $m$, $\{e_{k,i}\}, \enskip k=1,\dots,n, \enskip i=1,\dots,m$, introduced in equation (\ref{eq.6}), are i.i.d. with $E[e_{k,i}]=0 \enskip \text{and} \enskip \text{Var}(e_{k,i}) = \sigma^2$ on some probability space $(\Omega,\Sigma_\Omega,P_\Omega)$.
    
    \item[\textbf{A2.}] Let $(\mathcal{P}(Q),\rho)$ be the space of probability measures $\mathcal{P}(Q)$ with the Prohorov metric $\rho$, where $(Q,d)$ is a separable and compact metric space.
    
    \item[\textbf{A3.}] For $T>0$ and $n \in \mathbb{N}$, set $\tau^n=T/n$ and $t_k^n = k\tau^n$, $k=1,\dots,n$. Then, for each $i=1,\dots,m$, there exists $\bar{y}_{i} \in C(\mathcal{P}(Q),C([0,T]))$ such that $\bar{y}_{i}(t_k^n;P) = \bar{y}_{k,i}(P)$, for $P \in \mathcal{P}(Q)$, $k=1,\dots,n$.
    
    \item[\textbf{A4.}] There exists a measure $\mu$ on $[0,T]$ such that for all $f \in C([0,T])$, as $n \to \infty$,
    \begin{align*}
        \frac{1}{n} \sum_{k=1}^{n} f(t_k^n) = \int_0^T f(t)d\mu_n(t) \to \int_0^T f(t)d\mu(t).
    \end{align*}
    
    \item[\textbf{A5.}] The ``true" distribution $P_0 \in \mathcal{P}(Q)$ is the unique minimizer of $J_0(P)$ where
    
    \begin{align*}
        J_0(P) = \sigma^2 + \frac{1}{m} \sum_{i=1}^{m} \int_0^T (\bar{y}_i(t;P_0)-\bar{y}_i(t;P))^2 d\mu(t).
    \end{align*}
\end{itemize}

\begin{theorem}
    Under assumptions \emph{\textbf{A0-A5}}, define
    \begin{align*}
    J_{n,m}(\bm{Y};P) = \sum_{i=1}^m \sum_{k=1}^{n} (Y_{k,i} - \bar{y}_i(t_k^n;P))^2,
    \end{align*}
    there exists a set $A \in \Sigma_\Omega$ with probability one such that for fixed $m$ and for all $\omega\in A$,
    \begin{align*}
        \frac{1}{m n} J_{n,m}(\bm{Y};P)(\omega) \to J_0(P)(\omega)
    \end{align*}
    as $n \to \infty$, for each $P \in \mathcal{P}(Q)$. In addition, the convergence is uniform on $\mathcal{P}(Q)$.
    \label{thm:consis}
\end{theorem}

\begin{proof}
Let $P \in \mathcal{P}(Q)$ be fixed. We have
\begin{align}
    \frac{1}{m n} J_{n,m}(\bm{Y};P) &= \frac{1}{m n} \sum_{i=1}^m \sum_{k=1}^{n} (Y_{k,i} - \bar{y}_i(t_k^n;P))^2 \nonumber \\
    &= \frac{1}{m n} \sum_{i=1}^m \sum_{k=1}^{n} (\bar{y}_i(t_k^n; P_0) + e_{k,i} - \bar{y}_i(t_k^n;P))^2 \nonumber \\
    &= \frac{1}{m n} \sum_{i=1}^m \sum_{k=1}^{n} \bigg{[} e_{k,i}^2 + 2 \big{(} \bar{y}_i(t_k^n; P_0) - \bar{y}_i(t_k^n;P) \big{)} e_{k,i} + \big{(} \bar{y}_i(t_k^n; P_0) - \bar{y}_i(t_k^n;P) \big{)} ^2 \bigg{]} \nonumber \\
    &= \frac{1}{m n} \sum_{i=1}^m \sum_{k=1}^{n} e_{k,i}^2 \label{eq.9}\\
    & \enskip \enskip + \frac{2}{m n} \sum_{i=1}^m \sum_{k=1}^{n} \big{(} \bar{y}_i(t_k^n; P_0) - \bar{y}_i(t_k^n;P) \big{)} e_{k,i} \label{eq.10}\\
    & \enskip \enskip + \frac{1}{m n} \sum_{i=1}^m \sum_{k=1}^{n} \big{(} \bar{y}_i(t_k^n; P_0) - \bar{y}_i(t_k^n;P) \big{)} ^2 \label{eq.11}.
\end{align}

For the first term, (\ref{eq.9}), by the Strong Law of Large Numbers, we have
\begin{align*}
    P_\Omega(\{\omega\in \Omega \enskip | \enskip \frac{1}{m n} \sum_{i=1}^m \sum_{k=1}^{n} e_{k,i}^2 \to \sigma^2\}) = 1.
\end{align*}

For the second term, (\ref{eq.10}), let $\tilde{e}_{k,i}= \big{(} \bar{y}_i(t_k^n; P_0) - \bar{y}_i(t_k^n;P) \big{)} e_{k,i}$, by continuity of $\bar{y}_i(t;P)$ and compactness of $[0,T]$, we have

\begin{align*}
    E[\tilde{e}_{k,i}] &= \big{(} \bar{y}_i(t_k^n; P_0) - \bar{y}_i(t_k^n;P) \big{)} E[e_{k,i}] = 0.\\
    \text{Var}[\tilde{e}_{k,i}] &= \big{(} \bar{y}_i(t_k^n; P_0) - \bar{y}_i(t_k^n;P) \big{)} ^2 \enskip \text{Var}[e_{k,i}]\\ &= \sigma^2 \big{(} \bar{y}_i(t_k^n; P_0) - \bar{y}_i(t_k^n;P) \big{)} ^2\\
    & \leq \sigma^2 \sup_{t \in [0,T]} \big{(} \bar{y}_i(t; P_0) - \bar{y}_i(t;P) \big{)} ^2\\
    & < \infty.
\end{align*}
So we have $\text{Var}[\tilde{e}_{k,i}] \leq S_P$ for some $S_P$; therefore, we have
\begin{align*}
    \sum_{k=1}^{\infty} \frac{\text{Var}[\tilde{e}_{k,i}]}{k^2} \leq S_P \sum_{k=1}^{\infty} \frac{1}{k^2} < \infty.
\end{align*}
Hence, by Kolmogorov's Law of Large Numbers, we have
\begin{align*}
    P_\Omega(\{\omega\in \Omega \enskip | \enskip \frac{2}{m n} \sum_{i=1}^m \sum_{k=1}^{n} \big{(} \bar{y}_i(t_k^n; P_0) - \bar{y}_i(t_k^n;P) \big{)} e_{k,i} \to 0\}) = 1.
\end{align*}

For the third term, (\ref{eq.11}), by assumptions \textbf{A4}, \textbf{A5}, and continuity of $\bar{y}_i(t;P)$, we have
\begin{align*}
    \frac{1}{m n} \sum_{i=1}^m \sum_{k=1}^{n} \big{(} \bar{y}_i(t_k^n; P_0) - \bar{y}_i(t_k^n;P) \big{)} ^2 & \to \frac{1}{m} \sum_{i=1}^{m} \int_0^T \big{(} \bar{y}_i(t; P_0) - \bar{y}_i(t;P) \big{)} ^2 d\mu(t) \\
    & \enskip \enskip = J_0(P) - \sigma^2.
\end{align*}
Let %$A_P = \{\omega\in \Omega \enskip | \enskip \frac{1}{m n} \sum_{i=1}^m \sum_{k=1}^{n} e_{k,i}^2 \to \sigma^2\} \cap \enskip  \{\omega\in \Omega \enskip | \enskip \frac{2}{m n} \sum_{i=1}^m \sum_{k=1}^{n} \big{(} \bar{y}_i(t_k^n; P_0) - \bar{y}_i(t_k^n;P) \big{)} e_{k,i} \to 0\},$
\begin{align*}
    A_P &= \{\omega\in \Omega \enskip | \enskip \frac{1}{m n} \sum_{i=1}^m \sum_{k=1}^{n} e_{k,i}^2 \to \sigma^2\} \\
    & \enskip \enskip \enskip \cap \enskip  \{\omega\in \Omega \enskip | \enskip \frac{2}{m n} \sum_{i=1}^m \sum_{k=1}^{n} \big{(} \bar{y}_i(t_k^n; P_0) - \bar{y}_i(t_k^n;P) \big{)} e_{k,i} \to 0\},
\end{align*}
then $P_\Omega(A_P) = 1$ and for every $\omega \in A_P$, by adding the three terms above, we get
\begin{align*}
    \frac{1}{m n} J_{n,m}(\bm{Y};P)(\omega) \to \sigma^2 + 0 + J_0(P)(\omega) - \sigma^2 = J_0(P)(\omega).
\end{align*}

Next, we want to find $A$ with $P_\Omega(A)=1$ such that we get the convergence for every $\omega \in A$ and for all $P \in \mathcal{P}(Q)$. Let $A_1 = \{ \omega \in \Omega \enskip | \enskip \frac{1}{m n} \sum_{i=1}^{m} \sum_{k=1}^n |e_{k,i}| \to E[|e_{1,1}|]\}$.
% \begin{align*}
%     A_1 = \{ \omega \in \Omega \enskip | \enskip \frac{1}{m n} \sum_{i=1}^{m} \sum_{k=1}^n |e_{k,i}| \to E[|e_{1,1}|]\}.
% \end{align*}
By the Strong Law of Large Numbers, $P_\Omega (A_1) = 1$.

Similar to the idea in \cite{Banks:2012}, we let $A = A_1 \bigcap \big[\bigcap_{P \in \tilde{\mathcal{P}}_d(Q)} A_P \big]$.
% \begin{align*}
%     A = A_1 \bigcap \Bigg[\bigcap_{P \in \tilde{\mathcal{P}}_d(Q)} A_P \Bigg].
% \end{align*}
We take the intersection over a countable number of sets, each with probability one, so $P_\Omega(A)=1$. We must show that $A \subset A_P$ for all $P \in \mathcal{P}(Q)$ from which will follow the desired convergence for every $\omega \in A$ and for all $P \in \mathcal{P}(Q)$.
For all $\omega \in A$, we have $\omega \in A_1$. Choose $N_1$ such that for all $n \geq N_1$
\begin{align*}
    \frac{1}{m n} \sum_{i=1}^{m} \sum_{k=1}^n |e_{k,i}| < E[|e_{1,1}|] + 1.
\end{align*}
By Assumption $\bf{A3}$ and $\tilde{\mathcal{P}}_d(Q)$ being dense in $\mathcal{P}(Q)$, we choose $P_M \in \tilde{\mathcal{P}}_d(Q)$ such that for all $i=1,\dots,m$
\begin{align*}
    \sup_{t \in [0,T]}|\bar{y}_i(t;P) - \bar{y}_i(t;P_M)| < \frac{\epsilon}{4\big{(}E[|e_{1,1}|]+1\big{)}}.
\end{align*}
Then for all $\omega \in A$, we have $\omega \in A_{P_M}$, and therefore that $$\omega \in \{\omega\in \Omega \enskip | \enskip \frac{2}{m n} \sum_{i=1}^m \sum_{k=1}^{n} \big{(} \bar{y}_i(t_k^n; P_0) - \bar{y}_i(t_k^n;P_M) \big{)} e_{k,i} \to 0\}.$$ Choose $N_2$ such that for all $n \geq N_2$,
\begin{align*}
    \Bigg{|} \frac{2}{m n} \sum_{i=1}^m \sum_{k=1}^{n} \big{(} \bar{y}_i(t_k^n; P_0) - \bar{y}_i(t_k^n;P_M) \big{)} e_{k,i}\Bigg{|} < \frac{\epsilon}{2}.
\end{align*}
Then for $n \geq \max\{N_1,N_2\}$,
\begin{align}
    \Bigg{|} \frac{1}{m n} J_{n,m}(\bm{Y};P) - J_0(P) \Bigg{|} &\leq \Bigg{|} \frac{1}{m n} \sum_{i=1}^m \sum_{k=1}^{n} e_{k,i}^2 - \sigma^2 \Bigg{|} \label{eq.12} \\
    & \enskip + \Bigg{|} \frac{2}{m n} \sum_{i=1}^m \sum_{k=1}^{n} \big{(} \bar{y}_i(t_k^n; P_0) - \bar{y}_i(t_k^n;P) \big{)} e_{k,i} \Bigg{|} \label{eq.13} \\ 
    & \enskip + \Bigg{|} \frac{1}{m n} \sum_{i=1}^m \sum_{k=1}^{n} \big{(} \bar{y}_i(t_k^n; P_0) - \bar{y}_i(t_k^n;P) \big{)} ^2 \label{eq.14} \\
    & \enskip \enskip \enskip \enskip - \frac{1}{m} \sum_{i=1}^{m} \int_0^T \big{(} \bar{y}_i(t; P_0) - \bar{y}_i(t;P) \big{)} ^2 d\mu(t) \Bigg{|}. \nonumber
\end{align}

For the first term, (\ref{eq.12}), since $\omega \in A$, we have $$\omega \in \{\omega\in \Omega \enskip | \enskip \frac{1}{m n} \sum_{i=1}^m \sum_{k=1}^{n} e_{k,i}^2 \to \sigma^2\},$$ so $$\Bigg{|} \frac{1}{m n} \sum_{i=1}^m \sum_{k=1}^{n} e_{k,i}^2 - \sigma^2 \Bigg{|} \to 0.$$

For the second term, (\ref{eq.13}),
\begin{align*}
    & \Bigg{|} \frac{2}{m n} \sum_{i=1}^m \sum_{k=1}^{n} \big{(} \bar{y}_i(t_k^n; P_0) - \bar{y}_i(t_k^n;P) \big{)} e_{k,i} \Bigg{|} \\
    &= \Bigg{|} \frac{2}{m n} \sum_{i=1}^m \sum_{k=1}^{n} \big{(} \bar{y}_i(t_k^n; P_0) - \bar{y}_i(t_k^n;P) + \bar{y}_i(t_k^n;P_M) - \bar{y}_i(t_k^n;P_M) \big{)} e_{k,i} \Bigg{|} \\
    &\leq \Bigg{|} \frac{2}{m n} \sum_{i=1}^m \sum_{k=1}^{n} \big{(} \bar{y}_i(t_k^n; P_0) - \bar{y}_i(t_k^n;P_M) \big{)} e_{k,i} \Bigg{|} + \frac{2}{m n} \sum_{i=1}^m \sum_{k=1}^{n} \big{|} \big{(} \bar{y}_i(t_k^n; P_M) - \bar{y}_i(t_k^n;P) \big{)} \big{|} \big{|} e_{k,i} \big{|} \\
    &< \frac{\epsilon}{2} + \frac{2}{m n} \sum_{i=1}^{m} \sum_{k=1}^n \sup_{t \in [0,T]} \big{|} \bar{y}_i(t;P_M) - \bar{y}_i(t;P) \big{|} \big{|} e_{k,i} \big{|} \\
    &< \frac{\epsilon}{2} + 2 \Bigg{(} \frac{\epsilon}{4\big{(}E[|e_{1,1}|]+1\big{)}} \Bigg{)} \frac{1}{m n} \sum_{i=1}^{m} \sum_{k=1}^n \big{|} e_{k,i} \big{|} \\
    &< \frac{\epsilon}{2} + 2 \Bigg{(} \frac{\epsilon}{4\big{(}E[|e_{1,1}|]+1\big{)}} \Bigg{)} \big{(} E[|e_{1,1}|] + 1 \big{)} \\
    &= \frac{\epsilon}{2} + \frac{\epsilon}{2} = \epsilon.
\end{align*}

For the third term, (\ref{eq.14}), by assumptions \textbf{A4} and \textbf{A5}, we have
\begin{align*}
    \Bigg{|} \frac{1}{m n} \sum_{i=1}^m \sum_{k=1}^{n} (\bar{y}_i(t_k^n; P_0) - \bar{y}_i(t_k^n;P))^2 - \frac{1}{m} \sum_{i=1}^{m} \int_0^T (\bar{y}_i(t; P_0) - \bar{y}_i(t;P))^2 d\mu(t) \Bigg{|} \to 0.
\end{align*}
Therefore we have $\frac{1}{m n} J_{n,m}(\bm{Y};P)(\omega) \to J_0(P)(\omega)$ for every $\omega \in A$ and for all $P \in \mathcal{P}(Q)$.

%\vspace{0.1in}

Next, we want to show that the convergence is uniform on $\mathcal{P}(Q)$ for $\omega \in A$. First, we will show that the sequence of functions $\frac{1}{m n} J_{n,m}(\bm{Y};P)$, $n=1,\dots,\infty$ is equicontinuous as functions of $P$. Fix $\omega \in A$, let $\epsilon > 0$, and take $P \in \mathcal{P}(Q)$. By compactness of $[0,T]$ and Assumption $\bf{A3}$, there exists a $\delta > 0$ such that for all $\tilde{P} \in B_\delta (P)$, and for all $i=1,\dots,m$,
\begin{align*}
    \sup_{t \in [0,T]} |\bar{y}_i(t;P) - \bar{y}_i(t;\tilde{P})|& < \min \bigg{\{} \frac{\epsilon}{6\big{(}E[|e_{1,1}|]+1\big{)}} \enskip , \enskip \frac{\epsilon}{3 \big{(}\sup\limits_{t \in [0,T]} |\bar{y}_i(t; P_0)|\big{)}} \bigg{\}},\\
    \sup_{t \in [0,T]} |\bar{y}_i(t;P)^2 - \bar{y}_i(t;\tilde{P})^2|& < \frac{\epsilon}{3}.
\end{align*}
Also, since $\omega \in A$, we have $\omega \in A_1$. We choose $N_3$ such that $n \geq N_3$ implies 
\begin{align*}
    \frac{1}{m n} \sum_{i=1}^{m} \sum_{k=1}^n |e_{k,i}| < E[|e_{1,1}|] + 1.
\end{align*}
Then for $n \geq N_3$ and for all $\tilde{P} \in B_\delta (P)$,
\begin{align*}
    & \Bigg{|} \frac{1}{m n} J_{n,m}(\bm{Y};P) - \frac{1}{m n} J_{n,m}(\bm{Y};\tilde{P}) \Bigg{|}\\
    &= \Bigg{|} \frac{1}{m n} \sum_{i=1}^m \sum_{k=1}^{n} (\bar{y}_i(t_k^n; P_0) + e_{k,i} - \bar{y}_i(t_k^n;P))^2 - \frac{1}{m n} \sum_{i=1}^m \sum_{k=1}^{n} (\bar{y}_i(t_k^n; P_0) + e_{k,i} - \bar{y}_i(t_k^n;\tilde{P}))^2 \Bigg{|}\\
    &= \Bigg{|} \frac{1}{m n} \sum_{i=1}^m \sum_{k=1}^{n} \bigg{(} 2e_{k,i} + \bar{y}_i(t_k^n; P_0) - \bar{y}_i(t_k^n;P) - \bar{y}_i(t_k^n; \tilde{P}) \bigg{)} \bigg{(} \bar{y}_i(t_k^n; \tilde{P}) - \bar{y}_i(t_k^n;P) \bigg{)} \Bigg{|}\\
    & \leq \frac{1}{m n} \sum_{i=1}^m \sum_{k=1}^{n} 2 \big{|} e_{k,i} \big{|} \big{|} \bar{y}_i(t_k^n; \tilde{P}) - \bar{y}_i(t_k^n;P) \big{|}\\
    & \enskip + \frac{1}{m n} \sum_{i=1}^m \sum_{k=1}^{n} \big{|} \bar{y}_i(t_k^n; P_0) \big{|}  \big{|} \bar{y}_i(t_k^n; \tilde{P}) - \bar{y}_i(t_k^n;P) \big{|}\\
    & \enskip + \frac{1}{m n} \sum_{i=1}^m \sum_{k=1}^{n} \big{|} \bar{y}_i(t_k^n;P)^2 - \bar{y}_i(t_k^n; \tilde{P})^2 \big{|}\\
    & \leq \frac{2}{m n} \sum_{i=1}^m \sum_{k=1}^{n} \big{|} e_{k,i} \big{|} \bigg{(} \sup_{t \in [0,T]} \big{|} \bar{y}_i(t; \tilde{P}) - \bar{y}_i(t; P) \big{|} \bigg{)}\\
    & \enskip + \frac{1}{m n} \sum_{i=1}^m \sum_{k=1}^{n} \big{|} \bar{y}_i(t_k^n; P_0) \big{|}  \bigg{(} \sup_{t \in [0,T]} \big{|} \bar{y}_i(t; \tilde{P}) - \bar{y}_i(t; P) \big{|} \bigg{)}\\
    & \enskip + \frac{1}{m} \sum_{i=1}^{m} \sup_{t \in [0,T]} \big{|} \bar{y}_i(t; P)^2 - \bar{y}_i(t; \tilde{P})^2 \big{|}\\
    &< 2 \big{(} E[|e_{1,1}|] + 1\big{)} \Bigg{(}
    \frac{\epsilon}{6\big{(}E[|e_{1,1}|]+1\big{)}} \Bigg{)} \\
    & \enskip + \frac{1}{m n} \sum_{i=1}^m \sum_{k=1}^{n} \sup_{t \in [0,T]} \big{|} \bar{y}_i(t; P_0) \big{|} \bigg{(} \frac{\epsilon}{3 \big{(}\sup\limits_{t \in [0,T]} |\bar{y}_i(t; P_0)|\big{)}} \bigg{)}\\
    & \enskip + \frac{1}{m} \sum_{i=1}^{m} \sup_{t \in [0,T]} \big{|} \bar{y}_i(t; P)^2 - \bar{y}_i(t; \tilde{P})^2 \big{|}\\
    &< \frac{\epsilon}{3} + \frac{\epsilon}{3} + \frac{\epsilon}{3} = \epsilon.
\end{align*}
Therefore, the sequence of functions $\frac{1}{m n} J_{n,m}(\bm{Y};P)$, $n=1,\dots,\infty$ is equicontinuous as functions of $P$, and by the Arzela-Ascoli Theorem, $\frac{1}{m n} J_{n,m}(\bm{Y};P) \to J_0(P)$ uniformly on compact subsets of $\mathcal{P}(Q)$, and therefore we have uniform convergence on $\mathcal{P}(Q)$.
\qed
\end{proof}

\begin{theorem}
    Under assumptions \emph{\textbf{A0-A5}}, define
    \begin{align*}
    P_{n,m} = \arg \min_{P \in \mathcal{P}(Q)} J_{n,m}(\bm{Y};P) = \arg \min_{P \in \mathcal{P}(Q)} \sum_{i=1}^m \sum_{k=1}^{n} (Y_{k,i} - \bar{y}_i(t_k^n;P))^2,
    \end{align*}
    for a fixed $m$, $\rho(P_{n,m},P_0) \to 0$ with probability one as $n \to \infty$, where $\rho$ is the Prohorov metric.
\end{theorem}

\begin{proof}
Fix $\omega \in A$ (defined in Theorem \ref{thm:consis}). We know that $\frac{1}{m n} J_{n,m}(\bm{Y};P)(\omega) \to J_0(P)$ for every $P \in \mathcal{P}(Q)$ by the previous theorem. Let $\delta > 0$, and define the $\delta$-neighborhood of $P_0$ by 

\begin{align*}
    B_\delta(P_0) = \bigg{\{} \tilde{P} \enskip \bigg{|} \enskip \Bigg{|} \int_Q f(\bm{q}) d\tilde{P} - \int_Q f(\bm{q}) dP_0 \Bigg{|} < \delta, \enskip \text{for all} \enskip f \in C_b(Q) \bigg{\}}.
\end{align*}
where $C_b(Q) = \{f: Q \to \mathbb{R} \enskip | \enskip f \enskip \text{is bounded and continuous} \}$.
Then $B_\delta(P_0)$ is open in $\mathcal{P}(Q)$, and its complement $(B_\delta(P_0))^C$ is compact, since $\mathcal{P}(Q)$ is compact. Also, by assumption \textbf{A5}, there exists $\epsilon > 0$ such that $J_0(P) - J_0(P_0) > 2 \epsilon$ for all $P \in (B_\delta(P_0))^C$.\\
Also, by Theorem \ref{thm:consis}, there exists $N_4$ such that for $n \geq N_4$,
\begin{align*}
    \Bigg{|}\frac{1}{m n} J_{n,m}(\bm{Y};P)(\omega) - J_0(P) \Bigg{|} < \frac{\epsilon}{2}
\end{align*}
for all $P \in \mathcal{P}(Q)$.
Then for $n \geq N_4$ and $P \in (B_\delta(P_0))^C$,
\begin{align*}
    & \frac{1}{m n} J_{n,m}(\bm{Y};P) - \frac{1}{m n} J_{n,m}(\bm{Y};P_0) \\
    &= \frac{1}{m n} J_{n,m}(\bm{Y};P) - \frac{1}{m n} J_{n,m}(\bm{Y};P_0) + J_0(P) - J_0(P) + J_0(P_0) - J_0(P_0) \\
    &= \Bigg{(} \frac{1}{m n} J_{n,m}(\bm{Y};P) - J_0(P) \Bigg{)} + \Bigg{(} J_0(P) - J_0(P_0) \Bigg{)} + \Bigg{(} J_0(P_0) - \frac{1}{m n} J_{n,m}(\bm{Y};P_0) \Bigg{)}\\
    &> -\frac{\epsilon}{2} + 2 \epsilon - \frac{\epsilon}{2} = \epsilon > 0.
\end{align*}
So for all $P \in (B_\delta(P_0))^C$, we get
\begin{align*}
    J_{n,m}(\bm{Y};P) > J_{n,m}(\bm{Y};P_0).
\end{align*}
However, by definition of $P_{n,m}$, we have 
\begin{align*}
    J_{n,m}(\bm{Y};P_{n,m}) \leq J_{n,m}(\bm{Y};P_0).
\end{align*}
Therefore, $P_{n,m} \in B_\delta(P_0)$ for all $n \geq N_4$, and since $\delta$ was arbitrary, we get
\begin{align*}
    \rho(P_{n,m},P_0) \to 0.
\end{align*}
\qed
\end{proof}

\section{Approximation and Convergence}
\label{sec:4}

After establishing the consistency of the estimators $P_{n,m}$, followed by the estimates $\hat{P}_{n,m}$, using the realizations, defined as follows,
\begin{align*}
    \hat{P}_{n,m} = \arg \min_{P \in \mathcal{P}(Q)} J_{n,m}(\pmb{\mathscr{y}};P) = \arg \min_{P \in \mathcal{P}(Q)} \sum_{i=1}^m \sum_{k=1}^{n} (\mathscr{y}_{k,i} - \bar{y}_{k,i}(P))^2,
\end{align*}
we need to establish computational convergence since $P_{n,m}$ or $\hat{P}_{n,m}$ cannot be computed, and should be approximated numerically by $\hat{P}_{n,m,M}^N$, where
\begin{align*}
        \hat{P}_{n,m,M}^N = \arg \min_{P \in \mathcal{P}_M(Q)} J_{n,m}^N(\pmb{\mathscr{y}},P) = \arg \min_{P \in \mathcal{P}_M(Q)} \sum_{i=1}^m \sum_{k=1}^{n} (\mathscr{y}_{k,i} - \bar{y}^N_{k,i}(P))^2.
\end{align*}

The following Corollary, proved in \cite{Banks:2012}, is needed for the proof of Theorem \ref{Thm:compconv}.

\begin{corollary}
    Assume $(Q,d)$ is separable. Then $(Q,d)$ is compact if and only if $(\mathcal{P}(Q),\rho)$ is compact.
    \label{Cor:comp}
\end{corollary}

Consider the following assumptions,
\begin{itemize}[leftmargin=0.25in]
    \item[\textbf{B1.}] For all $n$, $m$, and $N$, the map $P \mapsto J_{n,m}^N(\pmb{\mathscr{y}};P)$ is a continuous map.
    
    \item[\textbf{B2.}] Let $(\mathcal{P}(Q),\rho)$ be the space of probability measures $\mathcal{P}(Q)$ with the Prohorov metric $\rho$. For any sequence of probability measures $P_M$, such that $\rho(P_M,P) \to 0$ in $\mathcal{P}(Q)$, $\bar{y}^N_{k,i}(P_M) \to \bar{y}_{k,i}(P)$ as $N,M \to \infty$ for $k=1,\dots,n, \enskip i=1,\dots,m$.
    
    \item[\textbf{B3.}] For all $P \in \mathcal{P}(Q)$ and for $k=1,\dots,n, \enskip i=1,\dots,m$, $\bar{y}_{k,i}(P)$ and $\bar{y}^N_{k,i}(P)$ are uniformly bounded.
\end{itemize}

\begin{theorem}
    Let $(Q,d)$ be a compact and separable metric space, and $(\mathcal{P}(Q),\rho)$ be the space of probability measures $\mathcal{P}(Q)$ with the Prohorov metric $\rho$. Let
    \begin{align*}
        \mathcal{P}_M(Q) = \{P \in \tilde{\mathcal{P}}_d(Q) \enskip | \enskip P=\sum_{j=1}^{M} p_j \delta_{\bm{q}_j}, \bm{q}_j \in \{\bm{q}_j\}_{j=1}^{M}\} \subset \tilde{\mathcal{P}}_d(Q).
    \end{align*}
    Then, under assumptions \emph{\textbf{B1-B3}}, there exists minimizers $\hat{P}_{n,m,M}^N$ such that
    \begin{align*}
        \hat{P}_{n,m,M}^N = \arg \min_{P \in \mathcal{P}_M(Q)} J_{n,m}^N(\pmb{\mathscr{y}};P) = \arg \min_{P \in \mathcal{P}_M(Q)} \sum_{i=1}^m \sum_{k=1}^{n} (\mathscr{y}_{k,i} - \bar{y}^N_{k,i}(P))^2.
    \end{align*}
    In addition, for a fixed $n$ and $m$, there exists a subsequence of $\hat{P}_{n,m,M}^N$ that converges to some $\hat{P}_{n,m}^*$ as $M,N \to \infty$, and $\hat{P}_{n,m}^*$ satisfies
    \begin{align*}
        \hat{P}_{n,m}^* = \arg \min_{P \in \mathcal{P}(Q)} J_{n,m}(\pmb{\mathscr{y}};P) = \arg \min_{P \in \mathcal{P}(Q)} \sum_{i=1}^m \sum_{k=1}^{n} (\mathscr{y}_{k,i} - \bar{y}_{k,i}(P))^2.
    \end{align*}
    \label{Thm:compconv}
\end{theorem}
\begin{proof}
For any fixed $n$ and $m$, by continuity of the map $P \mapsto J_{n,m}^N(\pmb{\mathscr{y}};P)$ for all $n$, $m$, and $N$, per assumption \textbf{B1}, and by Corollary \ref{Cor:comp}, we can conclude that $\hat{P}_{n,m,M}^N$ exist.

Since $\tilde{\mathcal{P}}_d(Q)$ is dense in $\mathcal{P}(Q)$, for $M=1,\dots,\infty$, we can construct a sequence of measures $P_M \in \mathcal{P}_M(Q) \subset \mathcal{P}(Q)$, such that $\rho(P_M,P) \to 0$ in $\mathcal{P}(Q)$, then by assumption \textbf{B2} and \textbf{B3}, for some constant $c$, we have
\begin{align*}
   \bigg{|} J_{n,m}^N(\pmb{\mathscr{y}};P_M) - J_{n,m}(\pmb{\mathscr{y}};P) \bigg{|} &= \Bigg{|} \sum_{i=1}^m \sum_{k=1}^{n} (\mathscr{y}_{k,i} - \bar{y}^N_{k,i}(P_M))^2 - \sum_{i=1}^m \sum_{k=1}^{n} (\mathscr{y}_{k,i} - \bar{y}_{k,i}(P))^2 \Bigg{|}\\
   &= \Bigg{|} \sum_{i=1}^m \sum_{k=1}^{n} \bigg{[} -2\mathscr{y}_{k,i} \bar{y}^N_{k,i}(P_M) + (\bar{y}^N_{k,i}(P_M))^2 + 2\mathscr{y}_{k,i} \bar{y}_{k,i}(P)\\
   & \enskip \enskip \enskip \enskip \enskip \enskip \enskip \enskip \enskip \enskip - (\bar{y}_{k,i}(P))^2 - \bar{y}^N_{k,i}(P_M) \bar{y}^N_{k,i}(P) + \bar{y}^N_{k,i}(P_M) \bar{y}^N_{k,i}(P) \bigg{]} \Bigg{|}\\
   &< \sum_{i=1}^m \sum_{k=1}^{n} \Bigg{|} 2\mathscr{y}_{k,i} - \bar{y}^N_{k,i}(P_M) - \bar{y}_{k,i}(P) \Bigg{|} \Bigg{|} \bar{y}_{k,i}(P) - \bar{y}^N_{k,i}(P_M) \Bigg{|}\\
   &< c \sum_{i=1}^m \sum_{k=1}^{n} \Bigg{|} \bar{y}_{k,i}(P) - \bar{y}^N_{k,i}(P_M) \Bigg{|} \to 0.
\end{align*}
So we get, $J_{n,m}^N(\pmb{\mathscr{y}};P_M) \to J_{n,m}(\pmb{\mathscr{y}};P)$ as $N,M \to \infty$.\\
Also, by definition, we have
\begin{align}
    J_{n,m}^N(\pmb{\mathscr{y}};\hat{P}_{n,m,M}^N) \leq J_{n,m}^N(\pmb{\mathscr{y}};P_M)
    \label{inequality}
\end{align}
for each $n$, $m$, and $N$, and for all $P_M \in \mathcal{P}_M(Q)$. \\
And, since $\mathcal{P}(Q)$ is compact, there exists a subsequence of $\hat{P}_{n,m,M}^N$ that converges to some $\hat{P}_{n,m}^*$ as $M,N \to \infty$.
So by taking the limit as $M,N \to \infty$ of (\ref{inequality}), for all $P \in \mathcal{P}(Q)$, we get
\begin{align*}
    J_{n,m}(\pmb{\mathscr{y}};\hat{P}_{n,m}^*) \leq J_{n,m}(\pmb{\mathscr{y}};P);
\end{align*}
therefore,
\begin{align*}
    \hat{P}_{n,m}^* = \arg \min_{P \in \mathcal{P}(Q)} J_{n,m}(\pmb{\mathscr{y}};P) = \arg \min_{P \in \mathcal{P}(Q)} \sum_{i=1}^m \sum_{k=1}^{n} (\mathscr{y}_{k,i} - \bar{y}_{k,i}(P))^2.
\end{align*}
\qed
\end{proof}

In practice, we fix sufficiently large values of $M$ and $N$ in order to achieve the preferred accuracy level. The choice of $N$ depends on the numerical framework used for approximation of $\bar{y}_{k,i}(P)$. In order to achieve the level of accuracy that is desired, we select $M$ nodes, $\{\bm{q}_j\}_{j=1}^{M}$, and determine the values of the weights $p_j$ at each node $\bm{q}_j$ by reducing the optimization problem to a standard constrained estimation problem over Euclidean $M$-space as follows
\begin{align*}
    \hat{P}_{n,m,M}^N &= \arg \min_{P \in \mathcal{P}_M(Q)} J_{n,m}^N(\pmb{\mathscr{y}};P) \\
    &= \arg \min_{P \in \mathcal{P}_M(Q)} \sum_{i=1}^m \sum_{k=1}^{n} (\mathscr{y}_{k,i} - \bar{y}^N_{k,i}(P))^2\\
    &= \arg \min_{P \in \mathcal{P}_M(Q)} \sum_{i=1}^m \sum_{k=1}^{n} \bigg{(} \mathscr{y}_{k,i} - E[h_k( x_{k,i}^N, \phi_{0,i}^N, u_{k,i}; \bm{q})] \bigg{)}^2 \\
    &= \arg \min_{P \in \mathcal{P}_M(Q)} \sum_{i=1}^m \sum_{k=1}^{n} \bigg{(}\mathscr{y}_{k,i} - \int_Q h_k( x_{k,i}^N, \phi_{0,i}^N, u_{k,i}; \bm{q}) dP\bigg{)}^2\\
    &= \arg \min_{\tilde{\bm{p}} \in \widetilde{\mathbb{R}^M}} \sum_{i=1}^m \sum_{k=1}^{n} \bigg{(}\mathscr{y}_{k,i} - \sum_{j=1}^M h_k(x_{k,i}^N, \phi_{0,i}^N, u_{k,i}; \bm{q}_j) p_j\bigg{)}^2\\
    &= \arg \min_{\tilde{\bm{p}} \in \widetilde{\mathbb{R}^M}} \big{|}\big{|}\pmb{\mathscr{y}}- \bm{H} \bm{D}_{\tilde{\bm{p}}}\big{|}\big{|}^2_F,
\end{align*}
where $\tilde{\bm{p}} = (p_1,\dots,p_M) \in \widetilde{\mathbb{R}^M} = \{\tilde{\bm{p}} \enskip | \enskip p_j \in \mathbb{R}^+ , \sum_{j=1}^M p_j = 1\}$. And $\parallel.\parallel_F$ is the Frobenius norm, $\pmb{\mathscr{y}}$ is the $n \times m$ matrix of the realizations $(\{\mathscr{y}_{k,i}\}_{k=1}^{n})_{i=1}^{m}$, $\bm{H}$ is the $n \times m M$ matrix of $$\bm{H} = \bigg{[}(\{h_k(x_{k,i}^N, \phi_{0,i}^N, u_{k,i}; \bm{q}_1)\}_{k=1}^{n})_{i=1}^{m} \enskip \dots \enskip (\{h_k(x_{k,i}^N, \phi_{0,i}^N, u_{k,i}; \bm{q}_M)\}_{k=1}^{n})_{i=1}^{m} \bigg{]},$$ and $\bm{D}_{\tilde{\bm{p}}}$ is $m M \times m$ matrix consisting of diagonal matrices as follows, $$\bm{D}_{\tilde{\bm{p}}} = \bigg{[} diag(p_1) \enskip diag(p_2) \enskip \dots \enskip diag(p_M) \bigg{]} ^T,$$ where $diag(a)$ is an $m \times m$ diagonal matrix with $a$ on the diagonals.

As $M$ increases, the optimization problem becomes ill-conditioned. As a result, $M$ must be chosen large enough such that the desired accuracy of computational approximation is attained, yet large numerical errors in solving the optimization problem, due to being poorly conditioned, are avoided.

\section{Abstract Parabolic Systems}
\label{sec:5}

Let $V$ and $H$ be Hilbert spaces with $V$ continuously and densely embedded in $H$, i.e. $V \hookrightarrow H$. Let $<.,.>_H$ and $|.|_H$ denote the inner product and norm on $H$, respectively, and let $\parallel.\parallel_V$ denote the norm on $V$. We identify $H$ with its dual $H^*$ to obtain the Gelfand triple $V \hookrightarrow H \hookrightarrow V^*$ and let $<\cdot,\cdot>_{V^*,V}$ denote the duality pairing between $V^*$ and $V$ induced by the dense and continuous embedding of $V$ in $H$ and $H$ in $V^*$. For each $N=1,2,\dots$, let $V^N$ be a finite-dimensional subspace of $V$, and for each $\bm{q} \in Q$, let $a(\bm{q};.,.) : V \times V \to \mathbb{C}$ be a sesquilinear form satisfying the conditions given below.

\begin{itemize}[leftmargin=0.25in]
    \item[\textbf{C1.}] \textbf{(Boundedness)} There exists a constant $\alpha_0$ such that for all $\psi_1,\psi_2 \in V$, we have
    \begin{align*}
        |a(\bm{q};\psi_1,\psi_2)| \leq \alpha_0 \parallel \psi_1 \parallel_V \parallel \psi_2 \parallel_V.
    \end{align*}
    \item[\textbf{C2.}] \textbf{(Coercivity)} There exists $\lambda_0 \in \mathbb{R}$ and $\mu_0>0$ such that for all $\psi \in V$, we have
    \begin{align*}
        a(\bm{q};\psi,\psi) + \lambda_0 |\psi|_H^2 \geq \mu_0 \parallel \psi \parallel_V^2.
    \end{align*}
    \item[\textbf{C3.}] \textbf{(Continuity)} For all $\psi_1,\psi_2 \in V$ and $\bm{q}, \tilde{\bm{q}} \in Q$, we have $$|a(\bm{q},\psi_1,\psi_2)-a(\tilde{\bm{q}},\psi_1,\psi_2)| \leq d(\bm{q},\tilde{\bm{q}}) \parallel \psi_1 \parallel_V \parallel \psi_2 \parallel_V.$$
    \item[\textbf{C4.}] \textbf{(Approximation)} For every $x \in V$, there exists $x^N \in V^N$ such that $\parallel x-x^N \parallel_V \to 0$ as $N \to \infty$.
\end{itemize}

Consider the following abstract parabolic system in weak form,
\begin{align}
    <\dot{x},\psi>_{V^*,V} + a(\bm{q};x,\psi) &= <\bm{B}(\bm{q})u,\psi>_{V^*,V}, \quad \psi \in V, \nonumber\\
    x(0) &= x_0, \label{eq.16}\\
    y(t) &= \bm{C}(\bm{q}) x(t) \nonumber,
\end{align}
where $x_0 \in H$, $u \in L^2([0,T],\mathbb{R}^{\nu})$ is the input to the system, $y \in L^2([0,T],\mathbb{R})$ is the output of the system, $\bm{B}(\bm{q}): \mathbb{R}^{\nu} \to V^*$, and $\bm{C}(\bm{q}): V \to \mathbb{R}$ are bounded linear operators.  Using standard variational arguments (see for example, \cite{Lions:1971}) it can be shown that the system (\ref{eq.16}) has a unique solution in $$\big{\{} \psi \enskip | \enskip \psi \in L^2([0,T],V), \enskip \dot{\psi} \in L^2([0,T],V^*) \big{\}} \subset C([0,T],H).$$  However, to convert the system given in (\ref{eq.16}) to discrete-time and to argue convergence of our finite-dimensional approximations, we will rely on a linear semigroup approach. 

Under the assumptions \textbf{C1} and \textbf{C2}, $a(\bm{q};.,.)$ defines a bounded linear operator $\bm{A}(\bm{q}):V \to V^*$ such that for $\psi_1, \psi_2 \in V$,
\begin{align*}
    <\bm{A}(\bm{q})\psi_1,\psi_2>_{V^*,V} = -a(\bm{q};\psi_1,\psi_2).
\end{align*}
As in \cite{Banks:1997,Banks:1989,Tanabe:1979}, the operator $\bm{A}(\bm{q})$ is regularly dissipative and restricted to $Dom(\bm{A}(\bm{q})) = \{\psi \in V \enskip | \enskip \bm{A}(\bm{q})\psi \in H\}$, which is the infinitesimal generator of a holomorphic or analytic semigroup, $\{e^{\bm{A}(\bm{q})t} \enskip | \enskip t \geq 0\}$, of bounded linear operators on $H$. In addition, this semigroup can be extended and restricted to be a holomorphic semigroup on $V^*$ and $V$, respectively (see details in \cite{Banks:1997,Tanabe:1979}). It follows that the system (\ref{eq.16}) can be written in a state space form as the evolution system with time invariant operators $\bm{A}(\bm{q})$, $\bm{B}(\bm{q})$, and $\bm{C}(\bm{q})$, as follows
\begin{align}
    \dot{x}(t) &= \bm{A}(\bm{q})x(t) + \bm{B}(\bm{q})u(t),\nonumber\\
    x(0) & = x_0, \label{eq.17}\\
    y(t) &= \bm{C}(\bm{q}) x(t) \nonumber.
\end{align}
with the mild solution of (\ref{eq.17}) given by the variation of constants formula as
\begin{align}
    x(t;\bm{q}) &= e^{\bm{A}(\bm{q}) t} x_0 + \int_0^t e^{\bm{A}(\bm{q})(t-s)} \bm{B}(\bm{q}) u(s) ds, \enskip t \geq 0, \label{eq.18}\\
    y(t;\bm{q}) &= \bm{C}(\bm{q}) x(t;\bm{q}) \nonumber.
\end{align}

Let $\tau > 0$ denote the length of the sampling interval and consider only zero-order hold inputs of the form $u(t) = u_{k-1}, \enskip t \in [(k-1)\tau,k\tau), \enskip k = 1,2,\dots$, and let $x_k = x(k\tau) \enskip \text{and} \enskip y_k = y(k\tau), \enskip k = 1,2,\dots$. Applying (\ref{eq.18}) on each sub-interval $[(k-1)\tau,k\tau]$, \enskip $k=1,2,\dots$, we obtain the discrete-time dynamical system given by
\begin{align}
    x_{k} &= \hat{\bm{A}}(\bm{q})x_{k-1} + \hat{\bm{B}}(\bm{q})u_{k-1}, \enskip k=1,2,\dots,\label{eq.19}\\
    y_{k} &= \hat{\bm{C}}(\bm{q}) x_k, \enskip k=1,2,\dots, \label{eq.20}
\end{align}
where $x_0 \in V$, $\hat{\bm{A}}(\bm{q}) = e^{\bm{A}(\bm{q}) \tau}$, $\hat{\bm{B}}(\bm{q}) = \int_0^\tau e^{\bm{A}(\bm{q}) s} \bm{B}(\bm{q}) ds$, and $\hat{\bm{C}}(\bm{q}) = \bm{C}(\bm{q})$.

We obtain a sequence of finite-dimensional systems approximating the discrete-time system given in (\ref{eq.19}) and (\ref{eq.20}) via Galerkin approximation (see, for example, \cite{Banks:1984}). For each $N \in \mathbb{Z}^+$, recall that $V^N$ is a finite-dimensional subspace of $V$ assumed to satisfy the approximation assumption \textbf{C4}. For each $\bm{q} \in Q$, let $\bm{A}^N(\bm{q})$ be the linear operator on $V^N$ obtained by restricting the form $a(\bm{q};.,.)$ to $V^N \times V^N$, i.e. for $\psi_1^N,\psi_2^N \in V^N$,
\begin{align*}
    <\bm{A}^N(\bm{q})\psi_1^N,\psi_2^N>_{V^*,V} = -a(\bm{q};\psi_1^N,\psi_2^N).
\end{align*}
Let $\pi^N: H \to V^N$ denote the orthogonal projection of $H$ onto $V^N$ and set $\bm{B}^N(\bm{q}) = \pi^N \bm{B}(\bm{q})$, where in this definition $\pi^N$ is understood to be the natural extension of the projection operator $\pi^N$ to $V^*$ from its dense subspace $H$. Define $\hat{\bm{A}}^N(\bm{q}) = e^{\bm{A}^N(\bm{q}) \tau}$, $\hat{\bm{B}}^N(\bm{q}) = \int_0^\tau e^{\bm{A}^N(\bm{q}) s} \bm{B}^N(\bm{q}) ds$, and $\hat{\bm{C}}^N(\bm{q}) x_k = \hat{\bm{C}}(\bm{q})$.  We then consider the sequence of approximating discrete-time dynamical systems given by 
\begin{align}
    x_{k}^N &= \hat{\bm{A}}^N(\bm{q})x_{k-1}^N + \hat{\bm{B}}^N(\bm{q})u_{k-1}, \enskip k=1,2,\dots,\label{eq.21}\\
    y_{k}^N &= \hat{\bm{C}}^N(\bm{q}) x_k^N, \enskip k=1,2,\dots, \label{eq.22}
\end{align}
with $x_0^N = \pi^N x_0\in V^N$.

Then under assumptions \textbf{C1-C4} together with the assumption that $Q$ is compact and continuity assumptions on the operators $\bm{B}(\bm{q})$ and $\bm{C}(\bm{q})$, using the Trotter-Kato theorem (see, for example, \cite{Banks:1988}), it can be argued that $\lim_{N \rightarrow \infty}\parallel x_k^N - x_k\parallel_V = 0$ and $\lim_{N \rightarrow \infty} | y_k^N - y_k | = 0$ for each $x_0 \in V$, and uniformly in $\bm{q}$ for $\bm{q} \in Q$ and $k \in \{1,2,\dots,K\}$, for any fixed $K \in \mathbb{N}^+$, where $x_k^N$ and $y_k^N$ are given by (\ref{eq.21}) and (\ref{eq.22}), respectively and $x_k$ and $y_k$ are given by (\ref{eq.19}) and (\ref{eq.20}), respectively.

We next show that the assumptions \textbf{B1-B3} in section \ref{sec:4} are satisfied. To establish that assumption \textbf{B1} holds, it is enough to show that for each $N$, $k$, and $i$ and
any sequence of probability measures $P_M \subset \mathcal{P}(Q)$ with $\rho(P_M,P) \to 0$ for some $ P \in\mathcal{P}(Q)$, we have $|\bar{y}_{k,i}^N(P_M) - \bar{y}_{k,i}^N(P)| \rightarrow 0$ as $M \rightarrow \infty$. For each $i=1,2,\dots,m$ and each $k=0,1,2,\dots,n$, it is not difficult to argue (for example via the Trotter-Kato theorem from linear semigroup theory, see for example, \cite{Pazy:1983}) that the map $\bm{q} \mapsto y_{k,i}^N(\bm{q})$ is continuous from $Q$ into $\mathbb{R}$, where for each $i=1,2,\dots,m$ and $\bm{q} \in Q$, $y_{k,i}^N(\bm{q})$ is given by (\ref{eq.21}) and (\ref{eq.22}). Since $Q$ was assumed to be compact, it follows that $y_{k,i}^N \in C_b(Q)$ and therefore from the definition of the Prohorov metric given in section \ref{sec:2}, that if $\rho(P_M,P) \to 0$ as $M \rightarrow \infty$, then $\bar{y}_{k,i}^N(P_M) \to \bar{y}_{k,i}^N(P)$ in $\mathbb{R}$ as $M \to \infty$, assumption \textbf{B1} is satisfied.

In order to show that the assumption \textbf{B2} is satisfied, from arguments given previously in this section we have that $\lim_{N \rightarrow \infty}\parallel x_{k,i}^N - x_{k,i}\parallel_V = 0$ and $\lim_{N \rightarrow \infty} | y_{k,i}^N - y_{k,i} | = 0$ for each $x_{0,i} \in V$, and uniformly in $\bm{q}$ for $\bm{q} \in Q$, $k=1,\dots,n, \enskip i=1,\dots,m$. We want to show that for any sequence of probability measures $P_M$, such that $\rho(P_M,P) \to 0$ in $\mathcal{P}(Q)$, and for $k=1,\dots,n, \enskip i=1,\dots,m$, as $N, M \to \infty$, we have 
\begin{align*}
    \bar{y}^N_{k,i}(P_M) = \int_Q y_{k,i}^N dP_M \to \bar{y}_{k,i}(P) = \int_Q y_{k,i} dP.
    \end{align*}
Let $\epsilon > 0$ be given and choose $N_0$ such that for $N \geq N_0$, we have $|y_{k,i}^N - y_{k,i}| < \epsilon/2$ on all of $Q$. Then for every $M$, since $P_M$ is a probability measure, we have $\int_Q |y_{k,i}^N - y_{k,i}| dP_M < \epsilon/2$. It follows that
\begin{align*}
    \Big{|} \int_Q y_{k,i}^N dP_M - \int_Q y_{k,i} dP \Big{|} 
     \leq \int_Q \big{|} y_{k,i}^N - y_{k,i} \big{|} dP_M + \Big{|} \int_Q y_{k,i} dP_M - \int_Q y_{k,i} dP \Big{|}
     \leq \frac{\epsilon}{2} + \frac{\epsilon}{2} = \epsilon,
\end{align*}
where the second term is made less than $\epsilon/2$ by using the fact that $\rho(P_M,P) \to 0$ and choosing $M$ sufficiently large. This establishes assumption \textbf{B2}.

Finally, the same reasoning used above to argue that for each $N$, $y_{k,i}^N(\bm{q})$ are uniformly bounded for $\bm{q} \in Q$ applies as well to $y_{k,i}(\bm{q})$. This together with the fact that $|y_{k,i}^N(\bm{q})-y_{k,i}(\bm{q})| \to 0$ uniformly in $\bm{q}$ for $\bm{q} \in Q$ are sufficient to conclude that assumption \textbf{B3} holds. 

\section{Application to the Transdermal Transport of Alcohol}
\label{sec:6}

In order to apply the results established in section \ref{sec:5} to the alcohol problem model, we must first rewrite the system (\ref{eq.1})-(\ref{eq.5}) in weak form, identify the feasible parameter space $Q$, the Hilbert spaces $H$ and $V$, the sesquilinear form $a(\bm{q};.,.) : V \times V \to \mathbb{C}$, and the operators $\bm{B}(\bm{q})$, and $\bm{C}(\bm{q})$. Then we must choose the approximating space $V^N$ and argue that the assumptions \textbf{C1-C4} are satisfied. 

The parameter space $Q$ is assumed to be a compact subset of $\mathbb{R}^+ \times \mathbb{R}^+$ with any $p$-metric denoted by $d_Q$. We set $V=H^1(0,1)$, $H=L^2(0,1)$, together with their standard inner products and norms, and therefore we have $V^*=H^{-1}(0,1)$. It is well-known that these three spaces form a Gelfand triple. For $\psi \in V$, integration by parts yields the weak form of (\ref{eq.1})-(\ref{eq.5}) to be
\begin{align*}
    <\dot{x}(t),\psi>_{V^*,V} + \int_0^1 q_1 \frac{\partial x}{\partial \eta}(t,\eta) \psi'(\eta) d\eta + x(t,0) \psi(t,0) = q_2 u(t) \psi(1).
\end{align*}
Consequently, for $\textbf{q} \in Q$, $u \in \mathbb{R}$, and $\varphi, \psi \in V$ we set
\begin{align*}
    a(\bm{q};\varphi,\psi) &= \int_0^1 q_1 \varphi'(\eta) \psi'(\eta) d\eta + \varphi(0) \psi(0),\\
    <\bm{B}(\bm{q})u,\psi>_{V^*,V} &= q_2 u\psi(1),\\
    \bm{C}(\bm{q})\psi&=\bm{C} \psi =\psi(0).
\end{align*}
Standard arguments involving the Sobolev Embedding Theorem (see, for example, \cite{Adams:2003}) can be used to argue that assumptions \textbf{C1-C3} are satisfied and clearly the operators $\bm{B}(\bm{q})$ and $\bm{C}(\bm{q})$ are continuous in the uniform operator topology with respect to $\bm{q} \in Q$. It follows from section \ref{sec:5} that in this case we have
\begin{align*}
    g_{k-1}(x_{k-1,i},u_{k-1,i};\bm{q}) &= g(x_{k-1,i},u_{k-1,i};\bm{q})\\
     &= \hat{\bm{A}}(\bm{q})x_{k-1,i} + \hat{\bm{B}}(\bm{q})u_{k-1,i}, \enskip k=1,2,\dots,n_i, \enskip i=1,\dots,m,\\
     h_k(x_{k,i}, \phi_{0,i}, u_{k,i}; \bm{q}) &= h(x_{k,i},u_{k,i};\bm{q})\\
     &= \hat{\bm{C}}(\bm{q}) x_{k,i}, \enskip k = 1,\dots,n_i, \enskip i = 1,\dots,m,
\end{align*}
where $\hat{\bm{A}}(\bm{q}) = e^{\bm{A}(\bm{q}) \tau}$, $\hat{\bm{B}}(\bm{q}) = \int_0^\tau e^{\bm{A}(\bm{q}) s} \bm{B}(\bm{q}) ds$, and $\hat{\bm{C}}(\bm{q}) = \bm{C}(\bm{q})$ with $\tau>0$ the length of the sampling interval.

For each $N=1,2,\dots$ let $V^N$ be the span of the standard linear splines (i.e. ``hat" or ``chapeau" functions) defined with respect to the uniform mesh $\{0,1/N,2/N,\dots,(N-1)/N,1\}$ on $[0,1]$. Then standard arguments for spline functions (see, for example, \cite{Schultz:1973}) can be used to argue that assumption \textbf{C4} is satisfied. If for each $i=1,2,\dots,m$, we then define $x_{k,i}^N$ and $y_{k,i}^N$ as in (\ref{eq.21}) and (\ref{eq.22}), the arguments at the end of section \ref{sec:5} yield that assumptions \textbf{B1-B3} are satisfied.

\subsection{Example 1: Estimation Based on Simulation Data}
\label{sec:6.1}

The transdermal transport of ethanol is modeled by a random partial differential equation with random parameter vector $\bm{q}=(q_1,q_2)$. We estimate the distribution of this parameter vector. To demonstrate the application of the approach developed in the previous sections, we simulated aggregate TAC data in MATLAB by assuming that the two parameters $q_1$ and $q_2$ are i.i.d. random variables from a $Beta(2,5)$ distribution and therefore that their joint cumulative distribution function (cdf) is the product of their marginal cdfs. Applying the results from the previous sections, we estimate the probability distribution of the parameter vector $\bm{q}=(q_1,q_2)$, and compare it to the ``true" distribution.

From equation (\ref{eq.6}), we have
\begin{align*}
    Y_{k,i} = \bar{y}_{k,i}(P_0) + e_{k,i}, \enskip k = 1,\dots,n_i, \enskip i = 1,\dots,m,
\end{align*}
where $m$ is the number of drinking episodes, $\bar{y}_{k,i}(P_0)$ is the observed aggregate TAC for the $i^{th}$ drinking episode at time step $k$, given the ``true" distribution $P_0$, the product of the cdfs of two independent $Beta(2,5)$ distributions, and $e_{k,i}$ is the random error, i.i.d. random variables from Normal distribution with mean $0$ and variance $1 \times 10^{-6}$, i.e. $N(0,1 \times 10^{-6}), \enskip k = 1,\dots,n_i, \enskip i = 1,\dots,m$.

The numerical scheme used to approximate the PDE model for the TAC observations is a linear spline-based Galerkin approximation method with $N+1$ equally spaced nodes ($N$ equal sub-intervals of $[0,1]$) (see \cite{Sirlanci:2019(1),Sirlanci:2017,Sirlanci:2019(2)}). Let $\bar{y}_{k,i}^N(P)$ be the approximation of $\bar{y}_{k,i}(P)$ using the Galerkin method, where $N$ denotes the level of the spatial discretization. And, let $\mathscr{y}_{k,i}$ be the realizations of the random variables $Y_{k,i}$. Let $n=\max n_i$, and let $y_{k,i} = 0$ for $n_i < k \leq n, \enskip i=1,\dots,m$, depicting the level of alcohol staying at $0$ after reaching $0$. We want to compute
\begin{align*}
    \hat{P}_{n,m,M}^N = \arg \min_{P \in \mathcal{P}_M(Q)} J_{n,m}^N(\pmb{\mathscr{y}},P) = \arg \min_{P \in \mathcal{P}_M(Q)} \sum_{i=1}^m \sum_{k=1}^{n} (\mathscr{y}_{k,i} - \bar{y}^N_{k,i}(P))^2,
\end{align*}
where $\pmb{\mathscr{y}}= (\{\mathscr{y}_{k,i}\}_{k=1}^{n_i})_{i=1}^{m}$. And, $\mathcal{P}_M(Q)$, an approximation to $\mathcal{P}(Q)$, is the computationally tractable set given by
\begin{align*}
    \mathcal{P}_M(Q) = \{P \in \tilde{\mathcal{P}}_d(Q) \enskip | \enskip P=\sum_{j=1}^{M} p_j \delta_{\bm{q}_j}, \bm{q}_j \in \{\bm{q}_j\}_{j=1}^{M}\},
\end{align*}
where
\begin{align*}
    \tilde{\mathcal{P}}_d(Q) = \{P \in \mathcal{P}(Q) \enskip | \enskip P=\sum_{j=1}^{M} p_j \delta_{\bm{q}_j}, \bm{q}_j \in Q_d, M \in \mathbb{N}, p_j \in [0,1] \cap \mathbb{Q}, \sum_{j=1}^{M} p_j =1\},
\end{align*}
the collection of all convex combinations of Dirac measures on $Q$ with nodes $\bm{q}_j \in Q_d$, where $Q_d = \{\bm{q}_j\}_{j=1}^{\infty}$, $M$ is the total number of nodes, and $p_j$ is the rational weight at the node $\bm{q}_j$. In the alcohol problem, since the parameter vector $\bm{q}$ is two-dimensional, we fix the nodes $\bm{q}_j = (q_{j_1},q_{j_2})$ as $M$ uniform meshgrid coordinates on $[0,1] \times [0,1]$.

From the results of the previous sections, we have
\begin{align*}
    \hat{P}_{n,m,M}^N = \arg \min_{\tilde{\bm{p}} \in \widetilde{\mathbb{R}^M}} \big{|}\big{|} \pmb{\mathscr{y}}- \bm{H} \bm{D}_{\tilde{\bm{p}}}\big{|}\big{|}^2_F,
\end{align*}
where $\pmb{\mathscr{y}}$ is the $n \times m$ matrix of the realizations $(\{\mathscr{y}_{k,i}\}_{k=1}^{n})_{i=1}^{m}$, $\bm{q}_j = (q_{j_1},q_{j_2})$ are the $M$ nodes of the uniform meshgrid, and $\tilde{\bm{p}} = (p_1,\dots,p_M) \in \widetilde{\mathbb{R}^M} = \{\tilde{\bm{p}} \enskip | \enskip p_j \in \mathbb{R}^+ , \sum_{j=1}^M p_j = 1\}$, where $p_j$ is the rational weight at the node $\bm{q}_j$. The $n \times m M$ matrix $\bm{H}$ is given by $$\bm{H} = \bigg{[}(\{h_k(x_{k,i}^N, \phi_{0,i}^N, u_{k,i}; \bm{q}_1)\}_{k=1}^{n})_{i=1}^{m} \enskip \dots \enskip (\{h_k(x_{k,i}^N, \phi_{0,i}^N, u_{k,i}; \bm{q}_M)\}_{k=1}^{n})_{i=1}^{m} \bigg{]},$$ where $h_k(x_{k,i}^N, \phi_{0,i}^N, u_{k,i}; \bm{q}_j)$ is the TAC evaluated at the node $\bm{q}_j = (q_{j_1},q_{j_2})$ and $\bm{D}_{\tilde{\bm{p}}}$ is $m M \times m$ matrix consisting of diagonal sub-matrices as follows, $$\bm{D}_{\tilde{\bm{p}}} = \bigg{[} diag(p_1) \enskip diag(p_2) \enskip \dots \enskip diag(p_M) \bigg{]} ^T.$$
We let $\phi_{0,i}^N = 0$, since we assume that there is no alcohol in the epidermal layer of the skin at time $t=0$. The resulting constrained optimization problem is solved using the routine FMINCON in MATLAB.

In order to prevent the estimation from extreme oscillations due to being ill-posed, we add a regularization term to our optimization problem. So we have
\begin{align*}
    \hat{\hat{P}}_{n,m,M}^N = \arg \min_{\tilde{\bm{p}} \in \widetilde{\mathbb{R}^M}} \big{|}\big{|} \pmb{\mathscr{y}}- \bm{H} \bm{D}_{\tilde{\bm{p}}}\big{|}\big{|}^2_F + w_1 \sum_{j_1,j_2} \big{|} p_{j_1+1,j_2} - p_{j_1,j_2} \big{|}^2 + w_2 \sum_{j_1,j_2} \big{|} p_{j_1,j_2+1} - p_{j_1,j_2} \big{|}^2,
\end{align*}
where $w_1$ and $w_2$ are the regularization parameters. As $w_1$ and $w_2$ increase, the regularization terms plays a stronger role. The critical step is to choose the regularization parameters, $w_1$ and $w_2$, properly, which is a combination of art and science.

Using MATLAB, we first simulated the TAC data which represents data that might be collected by the biosensor for an individual's drinking episode. To make our simulated TAC data as realistic as possible, we used BrAC data from three drinking episodes collected in the Luczak laboratory as the input to our model. Using these inputs, the simulated aggregate TAC is generated by first sampling $100$ TAC observations as longitudinal vectors through generating random samples of $q_1$ and $q_2$, i.i.d. random variables from a $Beta(2,5)$ distribution, then averaging them at each time point, and then adding noise.

After computing the simulated TAC outputs, we used our scheme, and numerically solved the resulting optimization problem to obtain an estimate for the cdf. To show the convergence of the estimated cdf of the parameter vector $\bm{q}=(q_1,q_2)$ to the ``true" distribution, which in the simulation case is the product of two $Beta(2,5)$ cdfs, we use the MCMC Metropolis Algorithm to generate random samples from the estimated distribution, and performed a generalized (2-dimensional) two-sample Kolmogorov-Smirnov test (KS-test). The null hypothesis was that the generated random samples from the estimated distribution and those from the ``true" distribution were drawn from the same distribution.

For the MCMC method, we used cubic spline interpolation to increase the resolution to a near-continuous distribution and used the Metropolis Algorithm with the standard uniform as the proposal distribution. For the Kolmogorov-Smirnov test, we used the Matlab function $kstest\rule{0.15cm}{0.15mm}2s\rule{0.15cm}{0.15mm}2d$ as described in \cite{Muir:2020}. The two-dimensional version of the KS-test used in this algorithm was developed by Peacock \cite{Peacock:1983}. The statistic, $D_{\hat{n}}$, used in the KS-test is the maximum of the absolute difference between the two empirical distributions of random samples generated by the estimated distribution and the ``true" distribution by considering all the ordering combinations that are possible.  In \cite{Peacock:1983}, Peacock argues that the two dimensional case may be analyzed in the same manner as the one dimensional case where the sample size is replaced by $$\hat{n} = \frac{n_1 n_2}{n_1+n_2},$$ where $n_1$  is the size of the sample generated by the estimated distribution, and $n_2$ is the size of the sample generated by the ``true" distribution. For this algorithm, $n_1$ and $n_2$ should be greater than $10$. Letting $Z_{\hat{n}} = \sqrt{\hat{n}}D_{\hat{n}}$, for large values of $n_1$ and $n_2$, $Z_{\hat{n}} \sim K$, where $K$ is a Kolmogorov random variable whose distribution is given by 
\begin{align*}
    F_K(z) = P(K \leq z) = 1 - 2\sum_{k=1}^\infty (-1)^{k-1}e^{-2k^2z^2}.
\end{align*}
 In \cite{Birnbaum:1952}, Birnbaum derived analytic expressions for the distribution of Kolmogorov's statistic for finite sample size, $ 2 \text{ and } 3$, and tabulated values for higher sample sizes. Using Monte Carlo simulation, Peacock demonstrated that one may adjust for small sample sizes via the fit asymptotic correction $$1 - \frac{Z_{\hat{n}}}{Z_{\infty}} = 0.53 \hat{n}^{-0.9},$$  Taking the null hypothesis to be that the two samples were drawn from the same distribution, the $p$-value is then well approximated by the expression $$P(K>Z_{\infty}) \approx 2 e^{-2 (Z_{\infty} - 0.5)^2},$$ the asymptotic distribution in the two-dimensional case (See details in \cite{Peacock:1983}).

We consider the case with $M=400$ and $N=128$ with regularization parameters $w_1=2 \times 10^{-3}$, and $w_2= 5 \times 10^{-5}$. In Figure (\ref{Fig.2}), we can see three different views of the estimated cdf and the ``true" cdf, which we recall is the product of two $Beta(2,5)$ cdfs.
\begin{figure}[H]
\centering
\includegraphics[width=4cm ,height= 3.2cm]{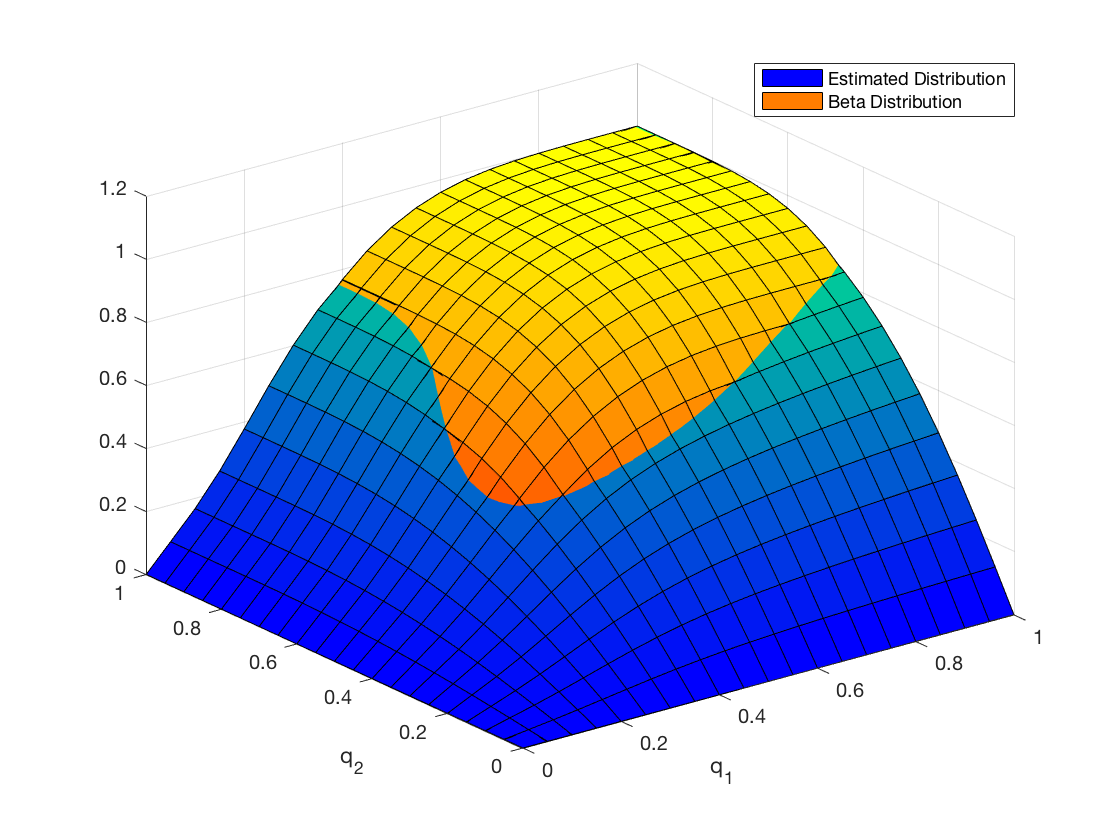}
\includegraphics[width=4cm ,height= 3.2cm]{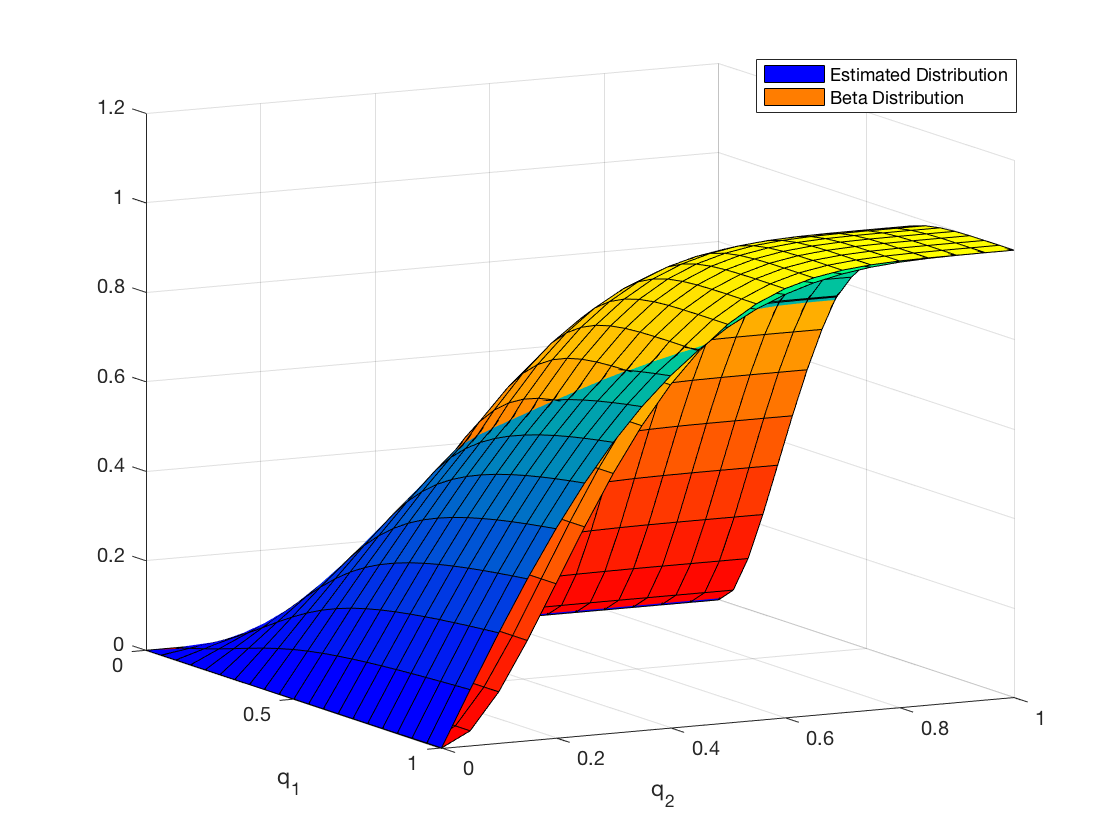}
\includegraphics[width=4cm ,height= 3.2cm]{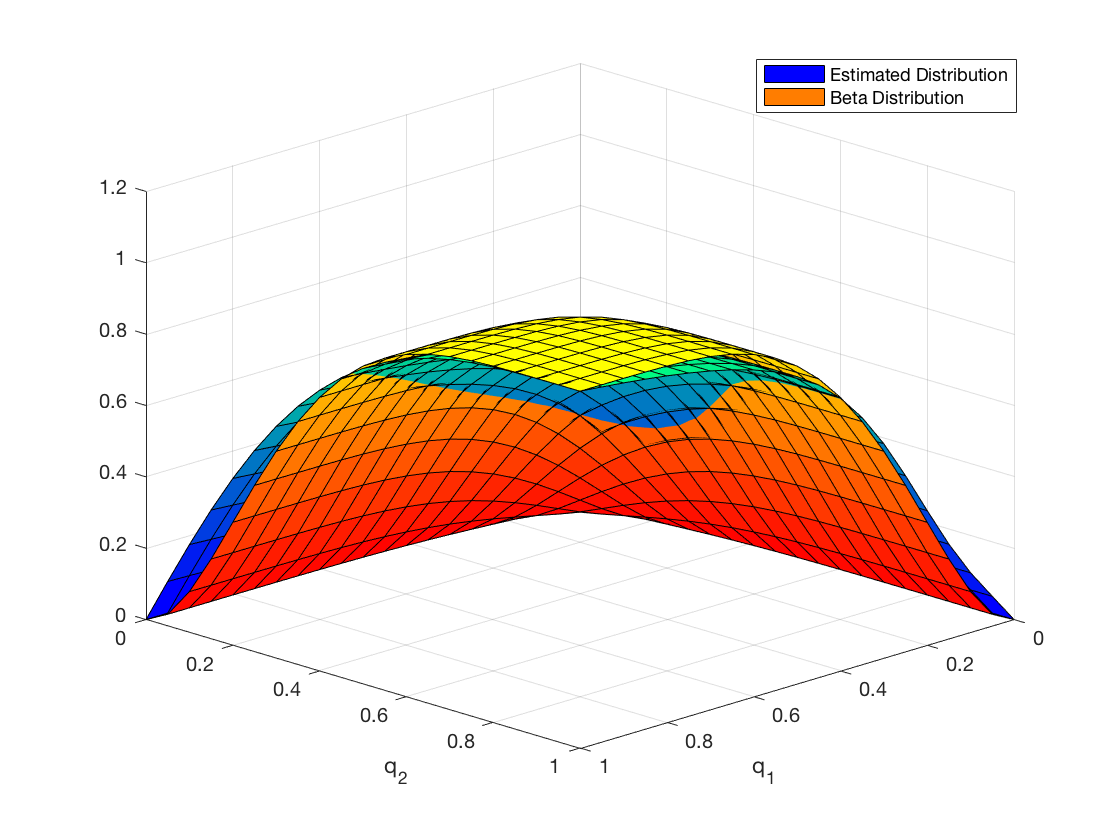}
\caption{Different views of the estimated distribution and the ``true" joint $Beta(2,5)$ distribution for $M=400$ and $N=128$ with regularization term}
\label{Fig.2}
\end{figure}
We can observe that our estimated distribution is reasonably ``close" to the ``true" distribution, which agrees with the averaged $p$-value of our hypothesis test in Table (\ref{Table.1}). To calculate the averaged $p$-value, we generated 500 samples from the estimated distribution using the MCMC algorithm and 500 samples from the ``true" distribution, and applied the two-dimensional Kolmogorov-Smirnov test. We repeated this $100$ times, and take the average of all the $p$-values that we get each time. The averaged $p$-value in Table (\ref{Table.1}) indicates that it is reasonable to not reject the null hypothesis, which states that the samples from the estimated distribution and the ``true" distribution are drawn from the same distribution.
\begin{table}[H]
\centering
    \caption{The averaged $p$-value for $M=400$ and $N=128$ with regularization parameters $w_1=2 \times 10^{-3}$, and $w_2= 5 \times 10^{-5}$}
    \label{Table.1}
    \begin{tabular}{lll}
    \hline\noalign{\smallskip}
    {$M$} & {$N$} & {Averaged $p$-value} \\
    \noalign{\smallskip}\hline\noalign{\smallskip}
    400 & 128 & 0.0586 \\
    \noalign{\smallskip}\hline
    \end{tabular}
\end{table}
Figure (\ref{Fig.3}) shows the histograms of the 50000 generated $q_1$ and $q_2$ samples. In this simulation case, since we assume that the parameters $q_1$ and $q_2$ are i.i.d. random variables from a $Beta(2,5)$ distribution, we compare the histograms of each parameter with the ``true" probability density function of $Beta(2,5)$ in one dimension.
\begin{figure}[H]
\centering
\includegraphics[width=4.6cm ,height= 3.32cm]{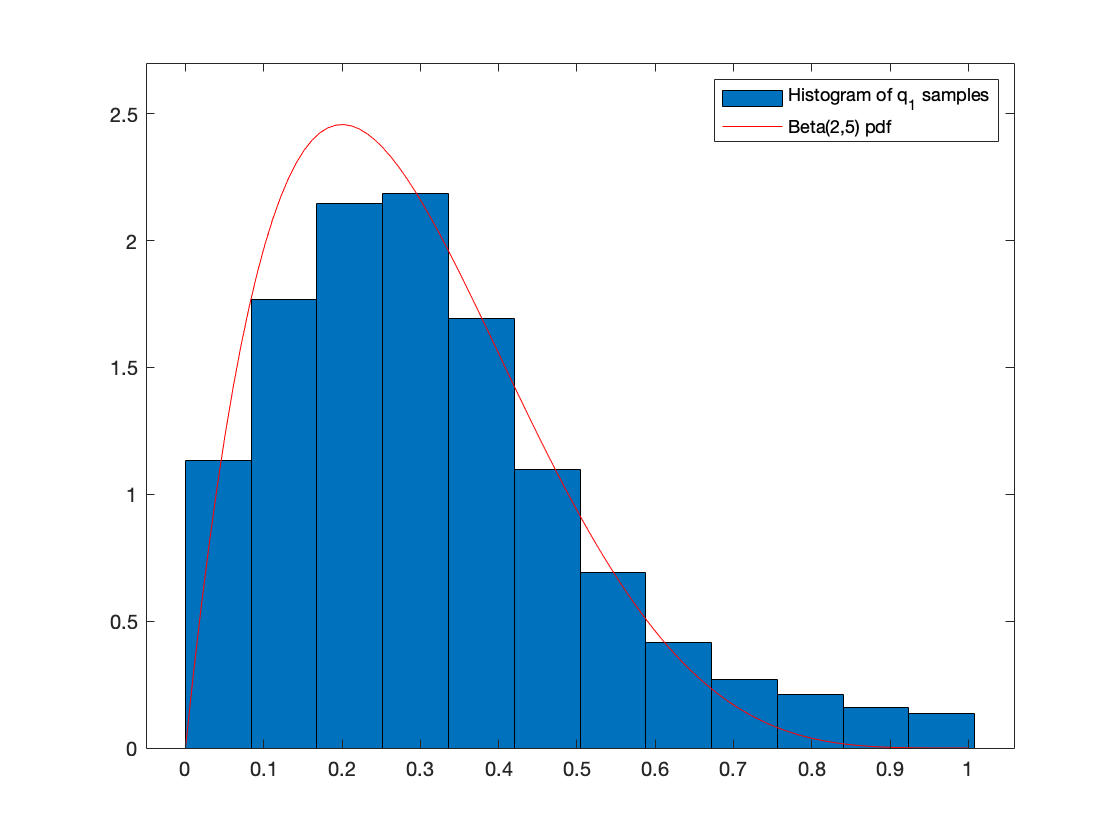}
\includegraphics[width=4.6cm ,height= 3.32cm]{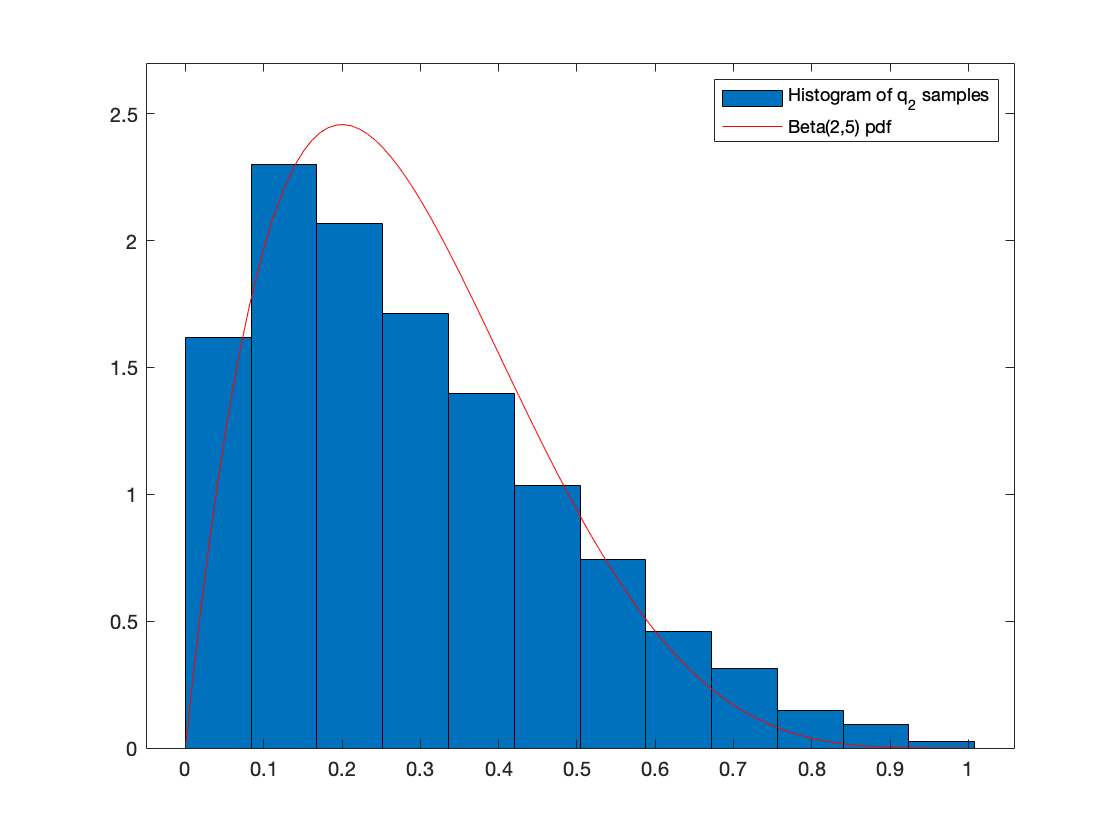}
\caption{Histogram of the generated $q_1$ samples (left) and $q_2$ samples (right) plus the ``true" probability density function of $Beta(2,5)$ (red)}
\label{Fig.3}
\end{figure}
Since these numerical results are based on the simulated data, we considered the ``true" distribution to be a simple case by assuming $q_1$ and $q_2$ i.i.d. random variables from a $Beta(2,5)$ distribution. The particular choice of this distribution is the scatter plot of samples of $q_1$ and $q_2$ for a particular set of drinking episode data given in \cite{Banks:2018}. In that paper, the parameters were obtained deterministically for BrAC and TAC measurements from 18 drinking episodes of different individuals. However, the choice of this distribution was mainly for demonstration purpose. We tested the algorithm on other distributions as well, and we obtained similar results. Next, we will apply the algorithm to actual human subject data collected in the Luczak laboratory. The nonparametric estimation method presented in this paper gives us the flexibility of not assuming any particular distribution type for the parameter vector $\bm{q}=(q_1,q_2)$, which in turn gives us the flexibility of not having to assume that $q_1$ and $q_2$ are i.i.d. when we are working with actual clinical or experimental data. By estimating the distribution of the parameter vector $\bm{q}=(q_1,q_2)$, we will then be able to sample from this distribution.

\subsection{Example 2: Estimation Based on Human Subjects Data}
\label{sec:6.2}

We considered two different datasets, \cite{Luczak:2015,Saldich:2020}, one collected with WrisTAS$^{\text{TM}}$7 biosensors and the other collected with SCRAM CAM$^{\textregistered}$ biosensors. For each dataset, we take $m = 9$ drinking episodes from different individuals. Applying the leave-one-out cross-validation (LOOCV) method, we partition the drinking episodes into the training set, which includes $8$ drinking episodes, and the test set, which includes $1$ drinking episode; we repeat this $9$ times. Each time, we get an estimation for the distribution of the parameter vector $\bm{q} = (q_1,q_2)$ using the training set. We generate $100$ random samples of $\bm{q}=(q_1,q_2)$ from the estimated distribution. Then we simulate $100$ TAC signals, using the BrAC input of the test set and the $100$ random samples of $\bm{q}$. We then compute the average of the TAC signals as an estimate of the ``true" TAC.

In order to estimate the accuracy of our model, we compute the normalized root-mean-square error (NRMSE) using the estimated TAC and the measured TAC, given the model complexity based on the number of nodes $M$ and the level of discretization $N$. In each round of the LOOCV method, since the comparison within the drinking episodes of the test sets between the measured TAC and the estimated TAC are in different scales for each test set, we use the following normalized root-mean-square error (NRMSE) for a means of comparison,
\begin{align*}
    \text{NRMSE}_{i} &= \frac{\text{RMSE}_{i}}{\max\limits_{k} \mathscr{y}_{k,i} - \min\limits_{k} \mathscr{y}_{k,i}},
\end{align*}
\noindent where
\begin{align*}
    \text{RMSE}_{i} &= \sqrt{\frac{1}{n} \sum_{k=1}^{n} (\mathscr{y}_{k,i} - \hat{\mathscr{y}}_{k,i})^2},
\end{align*}
\noindent with $\mathscr{y}_{k,i}$ the measured TAC for $i^{th}$ drinking episode at time step $k$, and $\hat{\mathscr{y}}_{k,i}$ the estimated TAC.

We compare various model complexities each having a different number of nodes $M$, and a different level of discretization $N$ by calculating the NRMSE$_{mean}$ for each model complexity,
\begin{align*}
    \text{NRMSE}_{mean} = \frac{1}{m} \sum_{i=1}^m \text{NRMSE}_{i}.
\end{align*}

For the first example, we consider the data collected using the WrisTAS7 alcohol biosensor. In Table (\ref{Table.2}), we include different model complexities with a different number of nodes $M$ and a different level of discretization $N$. We see that as the number of nodes $M$ and the level of discretization $N$ increase, the NRMSE$_{mean}$ decreases.
    \begin{table}[H]
    \centering
    \caption{Decrease in NRMSE$_{mean}$ for an increasing number of nodes $M$ and an increasing level of discretization $N$ for the dataset collected using the WrisTAS7 alcohol biosensor}
    \label{Table.2}
    \begin{tabular}{llll}
    \hline\noalign{\smallskip}
    {Model Complexity} & {$M$} & {$N$} & {NRMSE$_{mean}$} \\
    \noalign{\smallskip}\hline\noalign{\smallskip}
    1 & 4 & 2 & 0.8214\\
    2 & 9 & 2 & 0.4394\\
    3 & 16 & 4 & 0.2357\\
    4 & 25 & 4 & 0.1805\\
    5 & 36 & 8 & 0.1404\\
    6 & 49 & 16 & 0.1261\\
    7 & 64 & 32 & 0.1259\\
    8 & 81 & 128 & 0.1206\\
    9 & 225 & 128 & 0.1190\\
    10 & 400 & 128 & 0.1140\\
    \noalign{\smallskip}\hline
    \end{tabular}
    \end{table}
In Figure (\ref{Fig.4}), a plot with a y-axis on each side, we can see the decrease in NRMSE$_{mean}$ as the number of nodes $M$ and the level of discretization $N$ increase, referring to the model complexities in Table (\ref{Table.2}). We also see from the boxplots of NRMSE$_{i}$, $i=1,\dots,9$ for each model that the variance of NRMSE$_{i}$ also decreases with the increase in $M$ and $N$. Looking at the trend, we can see that the NRMSE$_{mean}$ does not decrease much after a certain point. Consequently, in practice, the choice of the number of nodes $M$ and the level of discretization $N$ presents a trade-off between the NRMSE$_{mean}$ and the computational cost.
\begin{figure}[H]
\centering
\includegraphics[width=6.7cm ,height= 5.66cm]{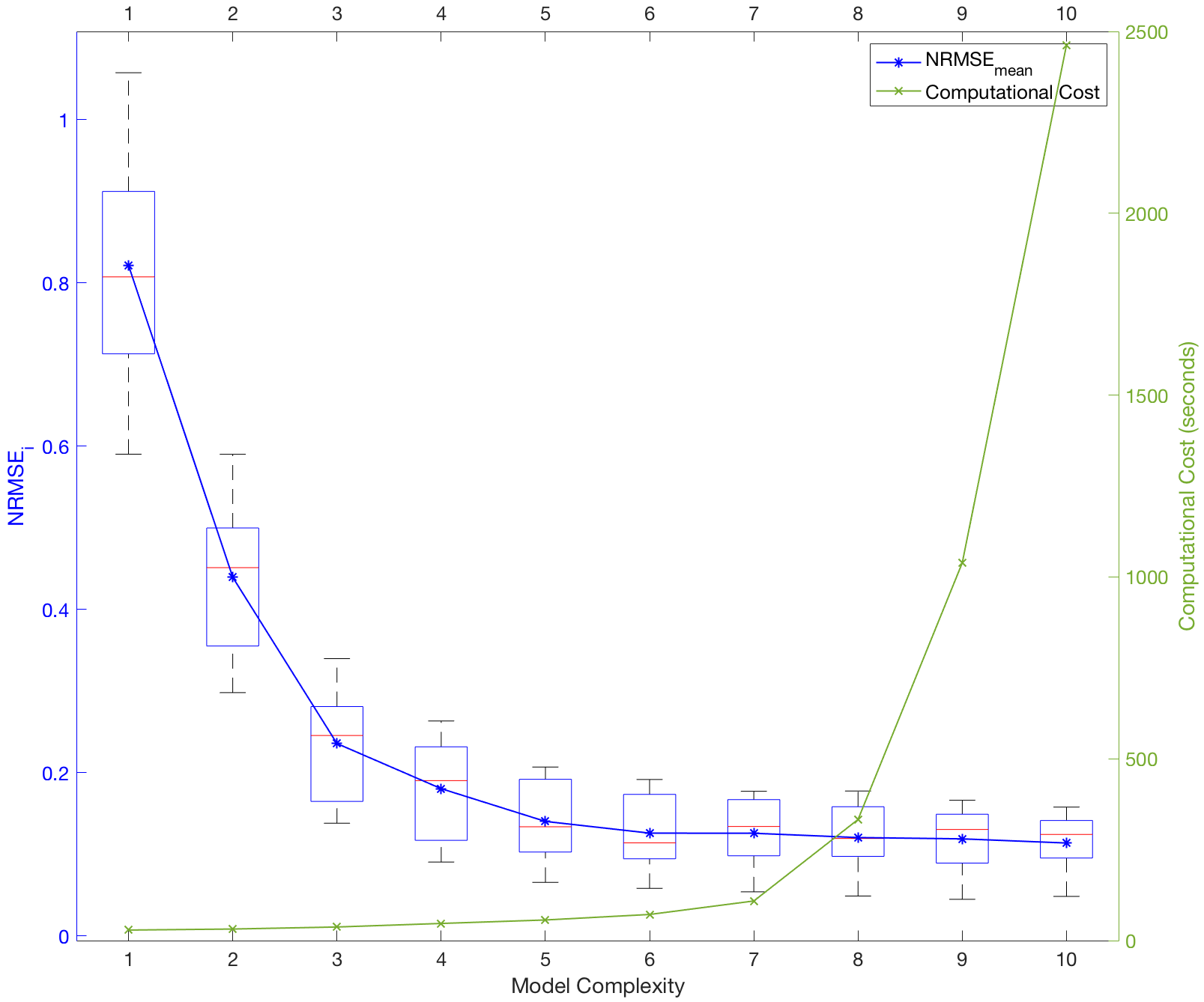}
\caption{Boxplot of NRMSE$_{i}$ given the complexity of the model as labeled in Table (\ref{Table.2}). A decrease in NRMSE$_{mean}$ as $M$ and $N$ increase (blue), and an increase in the computational cost in seconds (green) is observed for the dataset collected using the WrisTAS7 alcohol biosensor}
\label{Fig.4}
\end{figure}
Figure (\ref{Fig.5}) shows the measured TAC and the estimated TAC for the test set using the LOOCV, and the $95 \%$ simultaneous confidence band for the model complexity fixed with $M=400$ and $N=128$. We have an overall $1-\alpha$ confidence level with $\alpha = 0.05$ (i.e. $95 \%$ simultaneous confidence level) by using $1-\frac{\alpha}{n}$ confidence level at each time step according to the Bonferroni correction. The confidence interval at each time step is $ \Bar{x}_k \pm t_{99,\frac{\alpha}{2n}} \frac{s_k}{\sqrt{100}}$, where $\Bar{x}_k$ is the mean of the 100 sampled TACs at time step $k$, $s_k$ is the standard deviation of the 100 sampled TACs at time step $k$, and $t_{99,\frac{\alpha}{2n}}$ is the upper $\frac{\alpha}{2n}$ quantile of the $t$ distribution with $99$ degrees of freedom.
\begin{figure}[H]
\centering
\includegraphics[width=12.2cm ,height= 8.1cm]{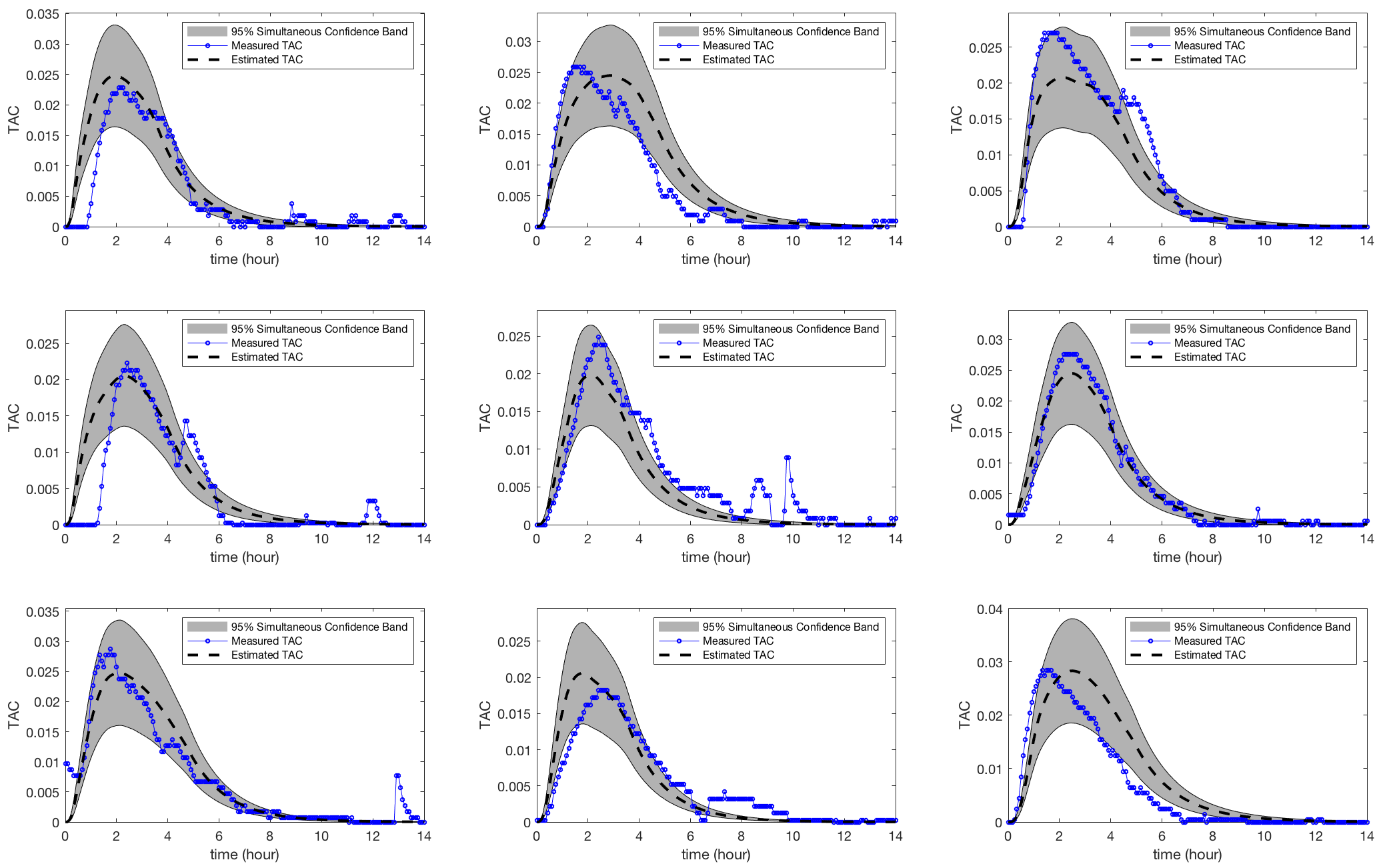}
\caption{The measured TAC, the estimated TAC, and the $95 \%$ simultaneous confidence band for 9 drinking episodes from the test set collected using the WrisTAS7 alcohol biosensor using the LOOCV}
\label{Fig.5}
\end{figure}
For the second example, we consider the dataset collected using the SCRAM alcohol biosensor. We can see from Table (\ref{Table.3}), and Figures (\ref{Fig.6}) and (\ref{Fig.7}) that we get a similar result.

    \begin{table}[H]
    \centering
    \caption{Decrease in NRMSE$_{mean}$ for an increasing number of nodes $M$ and an increasing level of discretization $N$ for the dataset collected using the SCRAM alcohol biosensor}
    \label{Table.3}
    \begin{tabular}{llll}
    \hline\noalign{\smallskip}
    {Model Complexity} & {$M$} & {$N$} & {NRMSE$_{mean}$} \\
    \noalign{\smallskip}\hline\noalign{\smallskip}
    1 & 4 & 2 & 0.5120\\
    2 & 9 & 2 & 0.2544\\
    3 & 16 & 4 & 0.1870\\
    4 & 25 & 4 & 0.1527\\
    5 & 36 & 8 & 0.1452\\
    6 & 49 & 16 & 0.1437\\
    7 & 64 & 32 & 0.1304\\
    8 & 81 & 128 & 0.1267\\
    9 & 225 & 128 & 0.1248\\
    10 & 400 & 128 & 0.1235\\
    \noalign{\smallskip}\hline
    \end{tabular}
    \end{table}
    
\begin{figure}[H]
\centering
\includegraphics[width=6.7cm ,height= 5.66cm]{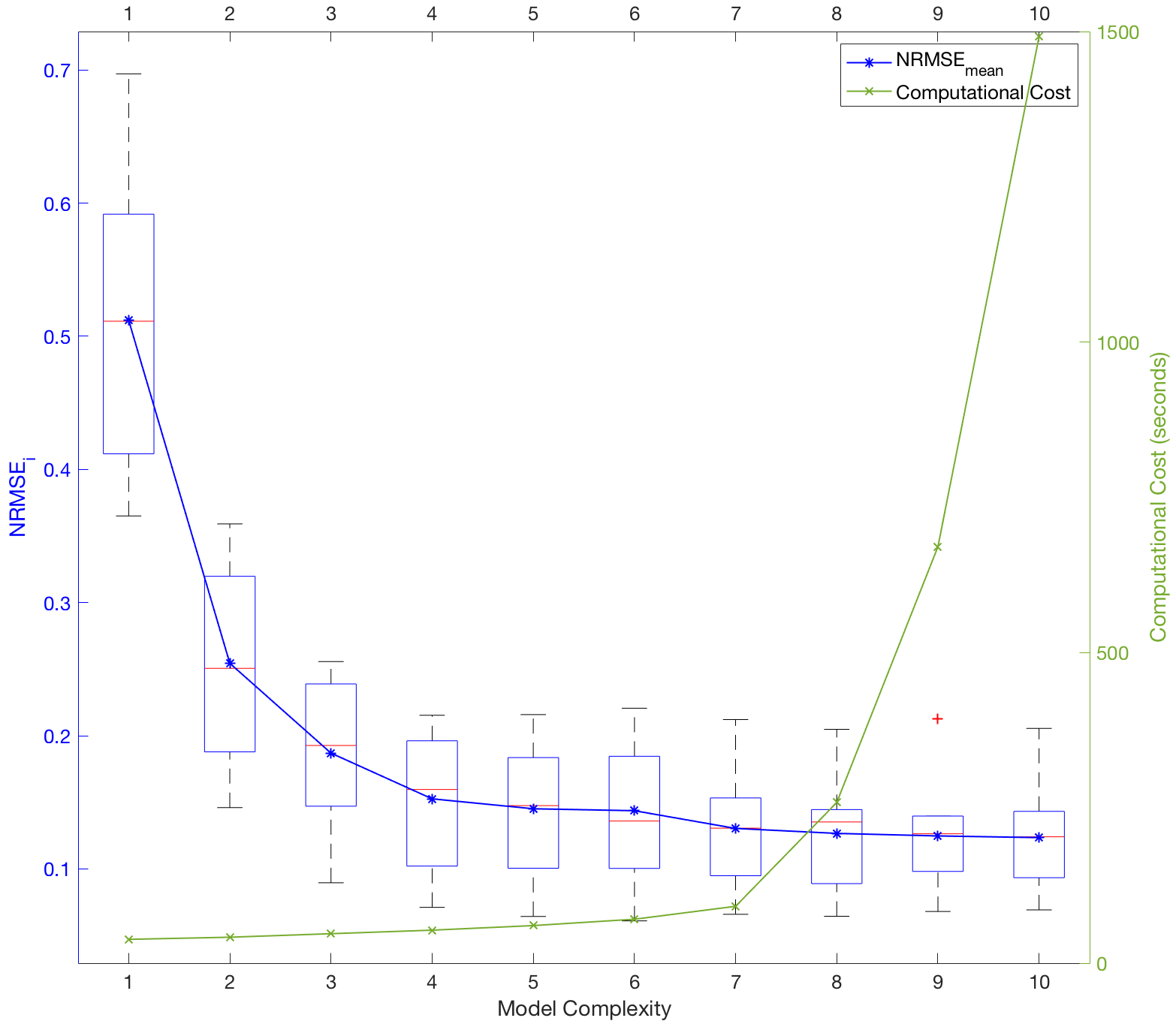}
\caption{Boxplot of NRMSE$_{i}$ given the complexity of the model as labeled in Table (\ref{Table.3}). A decrease in NRMSE$_{mean}$ as $M$ and $N$ increase (blue), and an increase in the computational cost in seconds (green) is observed for the dataset collected using the SCRAM alcohol biosensor}
\label{Fig.6}
\end{figure}

\begin{figure}[H]
\centering
\includegraphics[width=12.2cm ,height= 8.1cm]{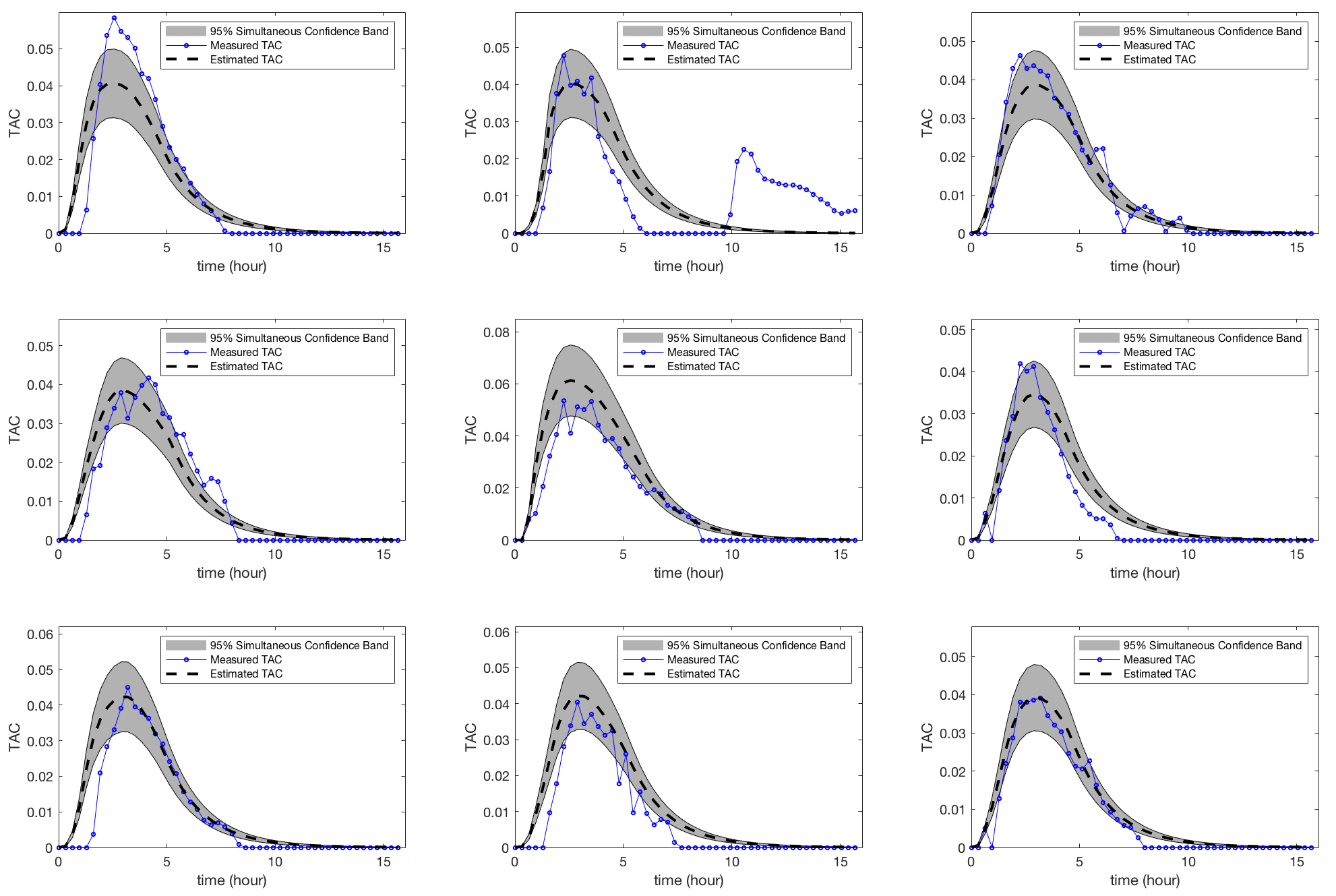}
\caption{The measured TAC, the estimated TAC, and the $95 \%$ simultaneous confidence band for 9 drinking episodes from the test set collected using the SCRAM alcohol biosensor using the LOOCV}
\label{Fig.7}
\end{figure}

\section{Discussion and Concluding Remarks}
\label{sec:7}

Our work here was motivated by a problem involving the development of a data analysis system for a transdermal alcohol biosensor. We assumed that each data point is an observation of the mean behavior plus a random error. Our results in section \ref{sec:6} demonstrated the efficacy of our approach in the case of both simulated data and human subjects data. In addition, in section \ref{sec:6.2}, we showed that our approach works for two different alcohol biosensors, the WrisTAS7 and the SCRAM devices.

In this study, we considered a nonparametric approach (i.e. not assuming any distribution form) to estimate the probability distribution of a random parameter vector in discrete-time (in general, infinite-dimensional) dynamical systems based on aggregate data from multiple subjects. Our results are related to, and in some sense an extension of, the Prohorov metric framework developed earlier in \cite{Banks:2012}. We obtained existence and consistency results for our estimator, and proved a convergence result for a finite-dimensional approximation framework for the case where the underlying dynamical system is infinite-dimensional and not directly amenable to numerical computations.  We provided a rather complete treatment of how our framework can be applied in the case of abstract parabolic systems. Our ultimate goal is to develop a generalized framework that is applicable to different types of dynamical systems whether finite-dimensional or infinite-dimensional and whether discrete-time or continuous-time. This generalization will extend our framework to apply to any type of dynamical system such as ordinary differential equations (ODE), partial differential equations (PDE), functional differential equations (FDE), or difference equations (DE).

Looking ahead, our primary motivation in pursuing this research is our interest in solving optimization and control problems involving random dynamical systems. Once we have estimated the distribution of the random parameters in the underlying dynamical system, which serves as the forward model, and introduced a population model, the objective then advances to estimating the input to the system based on the observation of the output for an individual subject (that is estimating BAC or BrAC from the observed TAC). In addition to providing an estimate of the input, the estimated distribution of the random parameters in the forward population model can then be used to obtain credible bands for the estimated input. The next step is currently ongoing.

We are also working on a similar approach to the one taken here that utilizes a maximum likelihood based mixed effects statistical model \cite{Davidian:2003,Stuart:2010} in place of the naive pooled data statistical model that formed the basis for the estimator developed in this study. The results we obtained and presented in this paper are based on the assumption that aggregate longitudinal data is what is available. However, if we assume that specific longitudinal data is available for each drinking episode, then a mixed effects model would be more appropriate. Finally, we are also investigating a Bayesian approach that yields posterior distributions for the random parameters in the underlying dynamical system wherein the prior distribution serves as regularization. Using the Bayesian approach, we can obtain credible bands, and we can compare the numerical results from the Bayesian approach to the results established here.

%\section{Section title}
%\label{sec:1}
%\subsection{Subsection title}
%\label{sec:2}
%as required. Don't forget to give each section
%and subsection a unique label (see Sect.~\ref{sec:1}).
%\paragraph{Paragraph headings} Use paragraph headings as needed.
%\begin{equation}
%\end{equation}

%\begin{figure}
%  \includegraphics{example.eps}
%\end{figure}
%\begin{figure*}
%  \includegraphics[width=0.75\textwidth]{example.eps}
%\end{figure*}
%\begin{table}
%\begin{tabular}{lll}
%\hline\noalign{\smallskip}
%\end{tabular}
%\end{table}

\begin{acknowledgements}
We thank the Luczak laboratory students and staff members, particularly Emily Saldich, for their assistance with data collection and management for the SCRAM biosensor. We also thank Dr. Tamara Wall for providing the data for the WrisTAS7 biosensor. %The authors would also like to thank the USC Women in Science and Engineering (WiSE) program for their support of one of the co-authors (LA).
\end{acknowledgements}
\section{Declarations}
\textbf{Funding:} This study was funded in part by the National Institute on Alcohol Abuse and Alcoholism (Grant Numbers: R21AA017711 and R01AA026368, S.E.L. and I.G.R.) and by support from the USC Women in Science and Engineering (WiSE) program (L.A.).\\ 
%\section*{Conflict of Interest}
%
\textbf{Conflict of Interest:} The authors declare that they have no conflicts of interest.\\
\textbf{Availability of Data and Material:} The data used in this study can be made available upon special request to the authors.\\
\textbf{Code Availability:} The codes used in this study can be made available upon special request to the authors. 



\end{document}